\documentclass[11pt,a4paper]{amsart}
\usepackage[margin=2.5cm]{geometry}

\usepackage[utf8]{inputenc}
\usepackage{amsmath, amssymb, amsfonts,amsthm,amsopn}
\usepackage{dsfont}
\usepackage{graphicx}

\usepackage{subcaption}
\usepackage{xltabular}
\usepackage{booktabs}\renewcommand{\arraystretch}{1.2}
\usepackage{titletoc}
\usepackage[section]{placeins}
\usepackage[shortlabels]{enumitem}
\usepackage{mathtools}

\usepackage{upgreek}
\usepackage{xcolor}
\usepackage{algorithm}
\usepackage{algpseudocode}
\algrenewcommand\algorithmicrequire{\textbf{Input:}}
\usepackage{thm-restate}
\usepackage{makecell}

\usepackage[breaklinks]{hyperref}
\hypersetup{
    colorlinks,
    linkcolor={red!40!black},
    citecolor={blue!50!black},
    urlcolor={blue!20!black}
}
\usepackage{orcidlink}
\usepackage[capitalise]{cleveref}

\DeclareMathOperator{\tr}{tr}

\DeclareMathOperator{\height}{height}
\DeclareMathOperator{\Wg}{Wg}

\DeclareMathOperator{\id}{id}

\DeclareMathOperator{\swap}{\textup{SWAP}}

\DeclareMathOperator{\Span}{span}
\DeclareMathOperator{\sym}{sym}

\newcommand{\ket}[1]{\mathinner{|#1\rangle}}

\newcommand{\ot}[0]{\otimes}
\newcommand{\one}[0]{\mathds{1}}

\renewcommand{\SS}{\mathcal{S}}

\newcommand{\R}{\mathds{R}}
\newcommand{\C}{\mathds{C}}

\newcommand{\E}{x}

\newcommand{\HH}{\mathcal{H}}

\newcommand{\UU}{\mathcal{U}}

\newcommand{\cdn}[0]{(\C^d)^{\otimes n}}

\newcommand{\expect}[1]{\langle #1 \rangle}

\newcommand{\norm}[1]{\left\|\,#1\,\right\|}       
\newcommand{\enorm}[1]{\norm{#1}_{\mathrm{2}}}      

\newtheorem{theorem}{Theorem}
\newtheorem*{theorem*}{Theorem}

\newtheorem{lemma}[theorem]{Lemma}

\newtheorem{corollary}[theorem]{Corollary}

\newtheorem{remark}[theorem]{Remark}

\newtheorem*{problem*}{Problem}
\newtheorem*{question*}{Question}

\newtheorem*{result*}{Result}

\newcommand{\nn}{\nonumber}

\newcommand{\fnote}[1]{}
\newcommand{\onote}[1]{}
\newcommand{\snote}[1]{}

\newcommand{\GL}{\textup{GL}}

\graphicspath{{code/figs/}}

\begin{document}

\title
[Second order cone relaxations for quantum Max Cut]
{Second order cone relaxations for quantum Max Cut}

\date{\today}

\author{Felix Huber, Kevin Thompson, Ojas Parekh, and Sevag Gharibian}

\address{
Institute of Theoretical Physics and Astrophysics, University of Gdańsk, Poland}
\email{felix.huber@ug.edu.pl}

\address{Sandia National Laboratories, Albuquerque, U.S.A.} 
\email{kevthom@sandia.gov}

\address{Sandia National Laboratories, Albuquerque, U.S.A.} 
\email{odparek@sandia.gov} 

\address{Department of Computer Science, and Institute for Photonic Quantum Systems, Paderborn University, Germany.}
\email{sevag.gharibian@upb.de}

\thanks{We thank
Ludovic Jaubert,
Titus Neupert,
Martin Plávala,
and
Claudio Procesi
for fruitful discussions and helpful comments.  This work is supported by a collaboration between the US DOE and other Agencies. 
FH was supported by the Foundation for Polish Science
through TEAM-NET (POIR.04.04.00-00-17C1/18-00),
Agence National de la Recherche (ANR-23-CPJ1-0012-01),
and the Region Nouvelle-Aquitaine.
This material is based upon work supported by
the U.S. Department of Energy, Office of Science, Office of Advanced Scientific Computing Research, Accelerated Research in Quantum Computing, Fundamental Algorithmic Research for Quantum Computing (FAR-QC) and Fundamental Algorithmic Research toward Quantum Utility (FAR-Qu).
This article has been authored by an employee of National Technology \& Engineering Solutions of Sandia, LLC under Contract No. DE-NA0003525 with the U.S. Department of Energy (DOE). The employee owns all right, title and interest in and to the article and is solely responsible for its contents. The United States Government retains and the publisher, by accepting the article for publication, acknowledges that the United States Government retains a non-exclusive, paid-up, irrevocable, world-wide license to publish or reproduce the published form of this article or allow others to do so, for United States Government purposes. The DOE will provide public access to these results of federally sponsored research in accordance with the DOE Public Access Plan https://www.energy.gov/downloads/doe-public-access-plan.
}

\begin{abstract}
Quantum Max Cut (QMC), also known as the quantum anti-ferromagnetic Heisenberg model, is a QMA-complete problem relevant to quantum many-body physics and computer science.  
Semidefinite programming relaxations have been
fruitful in designing theoretical
approximation algorithms for QMC,
but are computationally expensive for systems beyond tens of qubits.
We give a second order cone relaxation for QMC,
which optimizes over the set of mutually consistent three-qubit reduced density matrices.
In combination with Pauli level-$1$ of the quantum Lasserre hierarchy,
the relaxation achieves an approximation ratio of $0.526$ to the ground state energy.
Our relaxation is solvable on systems with hundreds of qubits and paves the way to computationally efficient lower and upper bounds
on the ground state energy of large-scale quantum spin systems.
\end{abstract}

\keywords{}

\maketitle
\setcounter{tocdepth}{1}

Approximating the ground energy or state of a local Hamiltonian is a fundamental problem in quantum physics.  While heuristic methods for providing upper bounds on the ground energy have been studied for many decades~\cite{doi:10.1126/science.231.4738.555, Or_s_2019, wu2023variational, doi:10.1126/science.aag2302},
techniques for providing lower bounds have bound mainly application in quantum chemistry~\cite{Mazziotti2007,Navascu_s_2008,10.1063/1.1360199, Barthel_2012,Baumgratz_2012, Fukuda2007, kull2022lower}.
Convex programs offer a systematic approach for relaxing both classical and quantum problems
and obtaining such lower bounds.
In particular, semidefinite programs (SDPs) are a natural fit for quantum mechanical problems as they allow modeling positivity of a state.
Recently, Quantum Max Cut (QMC) has established itself as a test bed for \emph{approximation algorithms} that provide rigorous bounds on the quality of the solution produced,
with approximation algorithms~\cite{
gharibian_et_al:LIPIcs:2019:11246,
https://doi.org/10.4230/lipics.tqc.2020.7,parekh2022optimal,
lee2022optimizing,
King2023improved,
lee2024improved}
based on noncommutative Lasserre SDP relaxations~\cite{
https://doi.org/10.4230/lipics.icalp.2021.102,
takahashi2023su2symmetric,
Watts2024relaxationsexact,
carlson2023approximation} at the forefront.
Quantum Max Cut problem is of particular interest because
it models antiferromagnetic quantum spin systems through an equivalence with the quantum Heisenberg Hamiltonian,
and it is a simple QMA-hard example of a 2-local Hamiltonian.

Given a graph $(G,E)$, quantum Max Cut problem asks for the ground state energy
of the Hamiltonian
 \begin{align}\label{eq:qmc_hamil}
  H &= - \sum_{(ij) \in E} \frac{1}{2}(\one - \swap_{ij})w_{ij},
 \end{align}
where the edge weights $\{w_{ij}\}_{(ij)\in E}$ are non-negative.
The swap interaction expands in terms of Pauli matrices as
$\swap = \frac{1}{2}(\one \ot \one + X \ot X + Y \ot Y + Z \ot Z)$.
 We can thus write the QMC Hamiltonian in Eq.~\eqref{eq:qmc_hamil} also as
 \begin{equation}
  H = - \frac{1}{4} \sum_{(i,j) \in E} (\one - X_i \ot X_j - Y_i \ot Y_j - Z_i \ot Z_j) w_{ij}.
 \end{equation}
 The QMC Hamiltonian tries to establish
 the maximally entangled singlet state $\tfrac{1}{\sqrt{2}}(\ket{01} - \ket{10})$ on every pair of vertices connected by an edge.
 However, the monogamy of entanglement (or more formally, the quantum de Finetti theorem)
 forbids the sharing of maximal entanglement with multiple partners.
 This induces frustration into the ground states, and consequentially,
 not all terms in the Hamiltonian can simultaneously be minimized.
 \begin{remark} Finding a ground state of the QMC Hamiltonian and the antiferromagnetic quantum Heisenberg model are equivalent since the Hamiltonians only differ in identity terms.  However, from an approximation perspective, the two problems differ.  Approximation guarantees for the Heisenberg model carry over to QMC but not vice versa.
 \end{remark}

With exact diagonalization, the largest systems
solved to date are 50 qubits using sublattice coding techniques~\cite{PhysRevB.100.155142, PhysRevE.98.033309}.
However, for disordered and large-scale systems such exact approaches fail
due to the growing Hilbert space dimension and the lack of spatial symmetries.
Variational algorithms can provide upper bounds, but work best on lattices in small dimensions.

Thus there is need for relaxations lower bounding the ground state energy.
Such approaches roughly fall into two classes:
a) quantum Lasserre or more general relaxations based on the NPA hierarchy~\cite{Navascu_s_2008}, which form moment matrices from the interaction algebra, and
b) marginal approximations, which replace the optimization over $\varrho$
   by a set of mutually consistent marginals~\cite{Mazziotti2007,10.1063/1.1360199, Barthel_2012,Baumgratz_2012, Fukuda2007, kull2022lower}.

\smallskip
\noindent {\bf Our contributions.}
The approach taken here is a mixture of the two types of approximation above:
we optimize over a set of reduced density matrices (quantum marginals)
that are mutually consistent as well as consistent with a Pauli level-1 quantum Lasserre moment matrix.
Conceptually,
this approach is similar to combining a classical Lasserre with Sherali-Adams hierarchy~\cite{sherali1990hierarchy},
with the difference that the local Sherali-Adams moment matrices are replaced by reduced density matrices.

First, we show that optimizing over mutually consistent 3-qubit reduced density matrices can be captured by a second order cone (SOC) program (Theorem~\ref{thm:SOC} in Section~\ref{sec:soc}).
SOC are specializations of semidefinite programs that can be solved more efficiently in practice.
In the process we show that a set of necessary constraints derived in~\cite{parekh2022optimal}
on the expectation values of 2-qubit quantum Max Cut terms are in fact sufficient (Theorem~\ref{thm:LM_PT} in Section~\ref{sec:3qb})
and can be expressed as a second order cone constraint.
This gives a concise description of the feasible space of expectation values for the three QMC terms on a triangle.

Second, we recall that the first level of the quantum Lasserre hierarchy yields a 0.498-approximation using product states~\cite{gharibian_et_al:LIPIcs:2019:11246},
and a better approximation using only the first level is unlikely~\cite{hwang2023unique}.
The second level of the quantum Lasserre hierarchy obeys certain monogamy of entanglement constraints~\cite{https://doi.org/10.4230/lipics.icalp.2021.102,parekh2022optimal},
described in the next section, enabling better approximations, with 0.595~\cite{lee2024improved} the current best.
However, the full generality of the second level of quantum Lasserre is not necessary to beat 0.498:
we show that the first level along with our second order cone constraints yields a 0.526-approximation (Theorem~\ref{thm:approx} in Section~\ref{sec:approx}).

On the practical side, we conduct a numerical study comparing our lower bounds with upper bounds obtained by state of the art methods on systems as large as 256 qubits (Section~\ref{sec:numerics}).
In contrast, SDP relaxations, including the recently introduced optimized SWAP hierarchies for quantum Max Cut, are typically practically impractical beyond tens of qubits ~\cite{takahashi2023su2symmetric,
Watts2024relaxationsexact}. 
Our numerics indicate that the relaxations we define are exact on Shastry-Sutherland models with small amounts of disorder (see \Cref{fig:SS_disorder}).
Similar to previously defined relaxations~\cite{takahashi2023su2symmetric}
we note that our hierarchy also ``mimics'' phase transitions in the Shastry-Sutherland model (see \Cref{fig:ShastrySutherland}).

Finally, we extend our relaxation to $k$-qubit marginals and local dimension $d > 2$ using the Weingarten formalism along with a new observation simplifying the approach.

We hope for our relaxations to become a new tool for physicists studying disordered systems and structureless local Hamiltonians. Our methods can inherently handle arbitrary interaction graphs, including irregular lattice models, such as those modeling defects (\Cref{fig:108qubits}).
These are typically hard to solve for variational methods due to lack of symmetry and locality.

\begin{figure}[tbp]
\includegraphics[width= 1\textwidth]{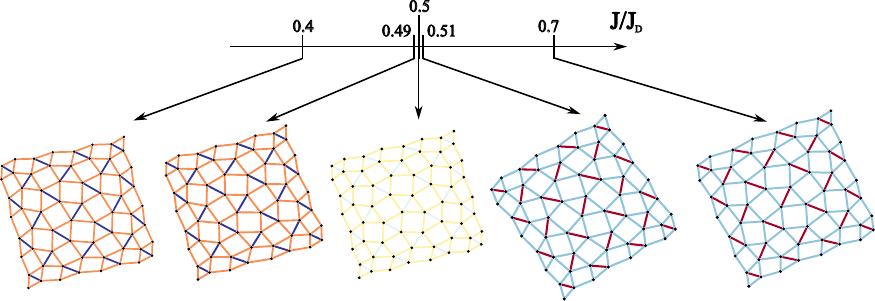}
 \caption{
 \label{fig:ShastrySutherland}
 Phase transitions between Dimer--Plaquette--Néel phases in the Shastry-Sutherland model calculated from our SOCP relaxation []\Cref{eq:relax2}].
 Dark blue edges indicate that the calculated objective along that edge is close to that of a singlet while dark red edges indicate that the objective is far from that of a singlet.
 Left: In the Dimer phase the ground state corresponds to a tensor product of singlet states along the blue edges.
 Center: The plaquette phase is an intermediate phase between these two extremes.
 Right: In the Néel phase the optimal state matches that of the square lattice,
 corresponding to aligned single qubit product states along the red edges. 
 While our relaxations mimic the behavior of the physical model, they do not precisely match it.
 For the physical model the transient plaquette phase occurs in a larger parameter regime with the phase transitions at different points.
}
\end{figure}

 \section{Main Results}

Write $x = \expect{\swap_{12}}$, $y = \expect{\swap_{13}}$, and $z = \expect{\swap_{23}}$.
The Lieb-Mattis and Parekh-Thompson~\cite{parekh2022optimal} conditions are:
\begin{align}
   0 \leq x + y + z &\leq 3\,, \label{eq:LM_I}\tag{LM}\\
   (x^2 + y^2 + z^2) - 2(xy + xz + yz) + 2(x + y + z) & \leq 3\,. \label{eq:PT_I}\tag{PT}
   \end{align}
We show that (see Section~\ref{sec:3qb}):
\begin{restatable}{thmA}{corLMPT}\label{thm:LM_PT}
The Lieb-Mattis~\eqref{eq:LM_I} and Parekh-Thompson triangle~\eqref{eq:PT_I} inequalities
form necessary and sufficient conditions for a real unitary-invariant trace one operator to represent a three-qubit quantum state.
\end{restatable}

\bigskip
A second order cone program (SOCP)  has form~\cite{LOBO1998193}
\begin{align}\label{eq:soc_prog}
 \underset{x}{\min} \quad & f^T x \nn\\
 \text{subject to}  \quad & \enorm{A_i x + b_i} \leq c_i^T x + d_i\,, \quad i=1,\dots, N\,,
\end{align}
where $\enorm{\cdot}$ is the Euclidean norm.
The variable $x\in\R^n$,
and the problem parameters are
$f \in \R^n$,
$A_i \in \R^{n_i-1 \times n}$,
$b_i \in \R^{n_i-1}$,
$c_i \in \R^n$, and
$d_i \in \R$.
Equality constraints of the form $A_i x = -b_i$ can be included by setting $c_i = d_i = 0$ in~\eqref{eq:soc_prog}.
Second order cone programs can be solved with interior point algorithms
and cover a larger class of problems than linear and quadratic programming.
Yet, they are also strictly weaker than semidefinite programming (SDP)~\cite{Fawzi2019}
and typically require less computational effort than similarly structured semidefinite programs~\cite{Alizadeh2003}.

The Lieb-Mattis and Parekh-Thompson conditions can be formulated as SOC constraints. It follows that (see Section~\ref{sec:soc}):

\begin{restatable}{thmA}{thmSOC}\label{thm:SOC}
 Quantum Max Cut is lower bounded by a second order cone program
 in $\binom{n}{2}$ variables.
 It corresponds to optimizing over mutually consistent 3-qubit reduced density matrices.
\end{restatable}
Here, {mutually consistent} means that the three-qubit density matrices (marginals) satisfy
that
$\tr_k(\varrho_{ijk}) = \tr_\ell(\varrho_{ik\ell})$ for all pairwise distinct $i,j,k,\ell \in \{1,\dots, n\}$.
Naively, these are constraints on a collection of positive-semidefinite $8\times 8$ matrices.
We use the $U^{\ot n}$-invariance and realness of the QMC Hamiltonian
to symmetry-reduce the positivity condition on a single three-qubit marginal
to a linear and a second order cone constraint.

\bigskip
We combine this with the quantum Lasserre hierarchy: the idea is to consider a (pseudo-) moment matrix whose entries $\Gamma(E,F)$ correspond to
expectations $\tr(\varrho E^\dag F)$. As with a moment matrix coming from a state $\varrho$,
one requires that $\Gamma \succeq 0$ as well as all consistency constraints arising from the indexing set;
e.g. from the relation $E^\dag F = K^\dag H$ follows the condition $\Gamma(E,F) = \Gamma(K,H)$.
At level $d$, $\Gamma$ is indexed by all Pauli strings of weight at most $d$.
The Pauli level-1 quantum Lasserre relaxation then corresponds to optimizing a $3n\times 3n$ psd matrix.
In combination with the second-order cone relaxation, one has an approximation guarantee (see Section~\ref{sec:approx}):
\begin{restatable}{thmA}{thm1Lasserre}\label{thm:1Lasserre}[see \Cref{thm:approx}]
The Pauli level-1 Lasserre hierarchy in combination with the second order cone program
for mutually consistent 3-qubit reduced density matrices achieves an approximation ratio of $0.526$.
\end{restatable}

\bigskip
The best known approximation for QMC achieves a better approximation factor of $0.595$~\cite{lee2024improved}.
However, the rounding algorithm for this and prior work uses second level of the NPA hierarchy, with respect to either Pauli operators (so-called quantum Lasserre) or the first level of a hierarchy with respect SWAP operators.
Either becomes intractable for instances as small as $20$ qubits.
We obtain a significant practical runtime improvement over all known approximation algorithms that approximate QMC to a factor at least $0.526$.

\bigskip
Our symmetry-reduction can somewhat straightforwardly be extended to $k$-partite relaxations and dimension $d\geq 3$.
We side-step the Weingarten formalism and simplify the expansion of the state by formulating positivity of an operator directly in terms of
its expectation values $\expect{\sigma} = \tr(\sigma A)$, $\sigma \in S_k$ (see Section~\ref{sec:kbody}). 

\begin{restatable}{thmA}{thmPOS}\label{thm:pos}
Let $A$ be unitary invariant and $R_\mu$ be irreducible orthogonal represenations.
Then
\begin{equation*}\label{eq:pos_in_irrep_simple}
 A \succeq 0 \quad \iff \quad
 \bigoplus_{\substack{\mu \vdash k \\ |\mu| \leq d}}
 \sum_{\sigma \in S_k} \langle \sigma \rangle  R_\mu( \sigma) \succeq 0\,.
\end{equation*}
\end{restatable}
This avoids the computation of characters of the symmetric group and Schur polynomials
in the usual expansion through the Weingarten formula.

\section{Proof ideas}

\noindent {\bf Theorem A: Lieb-Mattis and Parekh-Thompson conditions.}
Consider a relaxation of the quantum Max Cut problem that optimizes
over the set of mutually compatible reduced three-qubit density matrices.
A result by Eggeling
and Werner~\cite{PhysRevA.63.042111} shows that the any unitary-invariant three-qubit state can be block-diagonalized,
so that the positivity of $\varrho_{ijk}$ corresponds to the positivity in two subspaces (formally, irrep).
The first irrep is $1$-dimensional,
while the second is irrep is real and $2$-dimensional,
for which the positive semidefinite constraint can be formulated as a second order cone constraint.
The third (completely anti-symmetric) irrep of the symmetric group~$S_3$ vanishes and thus does not give a constraint on positivity.

Let us now see how the number of variables reduces from $3!=6$ to $3$:
first, the vanishing irrep corresponds to the identity $\id - (12) - (13) - (23) + (123) + (132) = 0$ for qubits.
Furethermore, the QMC Hamiltonian is real and as a consequence $\langle (123) \rangle = \langle (132)\rangle \in \R$.
This reduces the number of variables to three and
yields the LM and PT conditions for unitary-invariant operators of trace one.

\smallskip
\noindent {\bf Theorem B: Second-order cone relaxation.}
Considering all three-qubit subsystems, the consistent marginals relaxation consists of a
second order cone program (SOCP) in $\binom{n}{2}$ real variables. Here each variable to corresponds to an pair of vertices.
A key ingredient to this second order cone program is to directly use the expectation values $\expect{\swap_{ij}} = \tr(\varrho\swap_{ij})$
instead of coefficients for an operator basis.

\smallskip
\noindent {\bf Theorem~\ref{thm:1Lasserre}: Approximation algorithm.} Our approximation algorithms follows previous approaches,
including the best-known 0.595 approximation, based on producing a state that is a tensor product of 1- and 2-qubit states.
Being able to find a good product state solution is necessary in some regimes;
for example, product states are approximately optimal as input graphs grow dense~\cite{brandao2016product}.
Level 1 of the quantum Lasserre hierarchy is necessary for a good product state approximation,
and we add to this our second order cone constraints to form an SDP relaxation.
All previous approaches use the LM condition on a star,
while we can only enforce this on stars with two edges.
This presents difficulties in obtaining a good approximation, requiring a new ingredient in our analysis.
Ideally one would be able to show that a set of SDP relaxation values resulting in a bad approximation guarantee cannot occur as an optimal SDP solution for an actual input graph.
Our analysis attempts to mimic this by keeping track of the fraction of edges with low, medium, or high SDP value and allows these fractions to be chosen adversarially.

\smallskip
\noindent {\bf Theorem~\ref{thm:pos}: symmetry-reduction of k-body marginals}
The usual expansion of an invariant operator reads in terms of expectation values as
$\varrho = \sum_{\pi \in S_k} \tr(\varrho \pi) \pi^{-1} \Wg(d,n)$,
where $\Wg(d,n)$ is the Weingarten operator
$
 \Wg(d,k)
  = \frac{1}{k!} \sum_{\substack{\lambda \vdash k \\ |\lambda| \leq d}} \frac{\chi_\lambda(\id)}{s_{\lambda,d}(1)} p_\lambda
$.
Here $\chi_\lambda$ is the character of the symmetric group,
$s_{\lambda,d} = s_\lambda(1,\dots, 1)$ ($d$ times) the Schur polynomial,
and $P_\mu$ is the central Young projector projecting onto the isotypic component labeled by $\mu$.
Theorem~\ref{thm:pos} then follows from the fact that $\Wg(d,k)$ acts as a positive scalar
on each non-vanishing isotypic component.

\begin{figure}[tbp]
\begin{subfigure}{0.49\textwidth}
\centering
 \includegraphics[width= 0.9\textwidth]{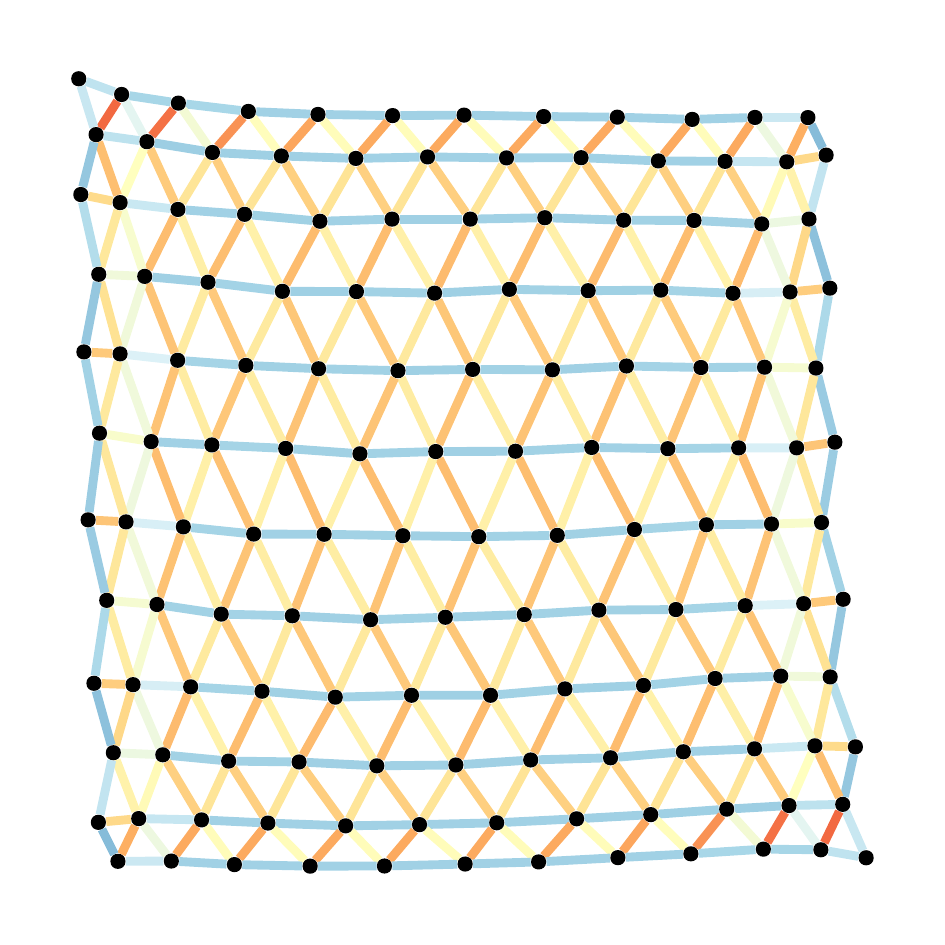}
\end{subfigure}
\begin{subfigure}{0.49\textwidth}
\centering
 \includegraphics[width= 0.9\textwidth]{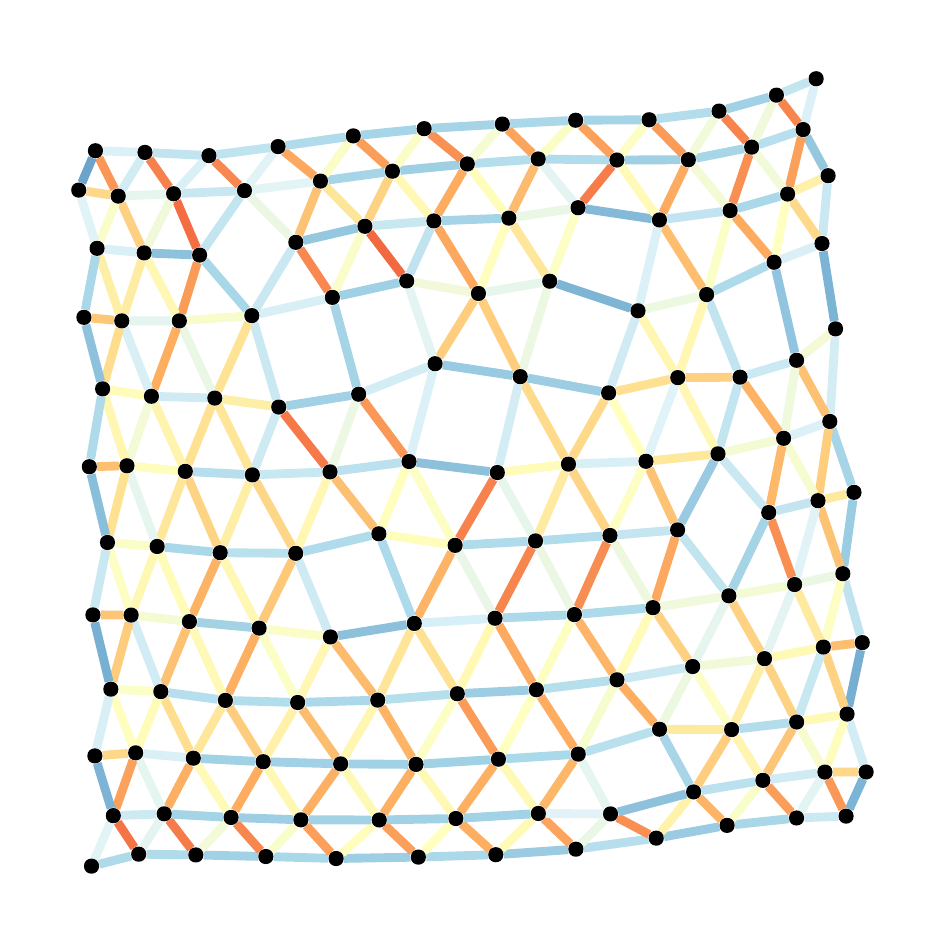}
\end{subfigure}
 \caption{
 \label{fig:108qubits}
 Approximations to quantum Max Cut on a 144-qubit lattice using the second order cone + Pauli level-1 relaxation.
 Shown are edge values that correspond to interaction terms in the Hamiltonian.
 Each subset of three vertices corresponds to a valid three-qubit reduced density matrix.
 A blue edge corresponds to the singlet state ($\langle \swap_{ij} \rangle= -1$)
 and satisfies its edge constraint, showing no frustration.
 A red edge correspond to a state orthogonal to the singlet ($\swap_{ij} \rangle = 1$) and indicates frustration.
 Compared to the lattice without defects where boundary effects dominate (right),
 lattice defects (left) curb to propagation of boundary effects into the interior, imposing local order.}
\end{figure}

\section{Discussion}

\subsection{Higher dimensions and larger marginals.}
Naturally, a similar approach can be taken for quantum Max Cut on higher spin systems:
Symmetry reduce the $k$-body marginals and demand their mutual consistency.
However, not all nice features of the three-qubit relaxation carry through:
For tri-partite system with $d\geq 2$ the partition $[1,1,1]$ does not vanish,
and consequently the relaxation requires an extra variable $\langle (123) \rangle = \langle (132) \rangle \in \R$.
The relaxation remains a second order cone program, now in four variables.
However, already starting from $k=4$ there are genuine semidefinite constraints appearing in the relaxation:
the $[3,1]$ partition of $S_4$ corresponds to a $3\times 3$ semidefinite constraint.
This makes the relaxations computationally less efficient.

Regarding a reduction in the number of variables,
in the case of $d=2$ (qubits),
this can be achieved with the Involution Basis
formed by $\{\sigma \in S_k\,|\, \sigma = \sigma^{-1}\}$
whose cardinality is the Catalan number $\binom{2n}{n} / (n+1)$~\cite{procesi2024specialbases2swapalgebras}.
In higher dimensions, the more general {\em Good Permutations} basis reduces the number of variables~\cite{procesi2020note}.


\subsection{Scaling.}
Most robust implementations solving SDP or SOCP rely on interior point methods.  While the number of total interior point iterations can be comparable for the two, the key difference in complexity comes from testing whether a potential solution is feasible.  In the SDP case, one must check that a matrix is positive semidefinite (PSD), while in the SOCP case, each SOC constraint can be checked in constant time in our case, and we have $O(n^3)$ such constraints, where $n$ is the number of qubits.  Checking feasibility of the second level of the (Pauli-based) quantum Lasserre hierarchy or the first level of the SWAP hierarchy entails checking that an $O(n^2) \times O(n^2)$ matrix is PSD.  Even if we augment our SOCP only with the first level of the quantum Lasserre hierarchy, as is necessary for \Cref{thm:1Lasserre}, the cost becomes checking whether an $O(n) \times O(n)$ matrix is PSD and checking $O(n^3)$ constraints of constant size.  Attempting to compute a Cholesky decomposition can be used to detect whether a matrix is PSD, and with this approach we get an $O(n^6)$ algorithm for checking feasibility for the quantum Lasserre or SWAP hierarchies, compared to $O(n^3)$ for our approach.  Of course this worst-case complexity may not reflect empirical behavior; however, our numerical study support the hypothesis that our SOCP relaxation is significantly more efficient than solving more conventional SDP relaxations.

\subsection{Outlook.}
Interestingly, the constraints corresponding to valid $3$-body marginals coincide exactly with monogamy constraints found by Ref.~\cite{parekh2022optimal} in the context of Lasserre relaxations.
This raises the question at which point the quantum Lasserre relaxations become exact.
Considering the Pauli Lasserre hierarchy alone,
the rank loop condition~\cite{doi:10.1137/090760155} in the Pauli hierarchy is met at level $\lceil n/2 \rceil$.
To accurately represent a k-qubit system, one thus requires to solve for at most level $\lceil k/2 \rceil$.
For hierarchies symmetrized to the Hamiltonian though,
it is not clear at what level convergence should happen and it is surprising that this relatively simple relaxation is exact on three qubits.
An interesting open direction is the understanding the minimal relations needed which coincide with the conditions for $k$-body marginals both for qubits and for $d\geq 3$.

The numerical results indicate that our relaxations can provide exact ground state energies
for simple models with disorder and can also detect phase transitions.
This adds to the growing body of evidence that convex relaxations can be powerful tools for physics
and in particular for perturbation theory~\cite{hastings2022perturbation}.
Our tools are particularly applicable to numerical study since second order cone programs
are more efficient to solve than generic semi-definite programs.

\section{Concepts}

\subsection{Notation}
We denote expectation values of observables $A$ on a state
$\varrho$ by
$\expect{A} = \tr(A \varrho)$.
Unless otherwise specified, we work with $n$-qubit systems on $(\C^2)^{\otimes n}$.
The Pauli matrices are denoted by $X,Y,Z$;
when acting on a subsystem $i$ as $X_i,Y_i,Z_i$.
The variables $x_{ij} \in \R$ approximate expectation values $\langle \swap_{ij} \rangle$ of swap operators $\swap_{ij}$.

\subsection{Permutations}

The swap operator exchanges the two tensor factors of $\C^d \ot \C^d$ as
\begin{equation}
\swap  \ket{v} \ot \ket{w} = \ket{w} \ot \ket{v}\,.
\end{equation}
We use $\swap_{ij}$ to denote the action of $\swap$ on qubits $i$ and $j$, i.e. the map
\begin{align}
  \ket{v_1} \ot \dots \ot
        \ket{v_{\underline{i}}} \ot \dots \ot
        \ket{v_{\underline{j}}} \ot \dots \ot
        \ket{v_n} 
 \mapsto
   \ket{v_1} \ot \dots \ot
        \ket{v_{\underline{j}}} \ot \dots \ot
        \ket{v_{\underline{i}}} \ot \dots \ot
        \ket{v_n}.
\end{align}
Note that for $d=2$,
\begin{align}
  \swap_{ij}=\frac{1}{2}\left(I_i\otimes I_j+X_i\otimes X_j+Y_i\otimes Y_j+Z_i\otimes Z_j\right).\label{eqn:SWAP}
\end{align}

Since $\swap_{ij}$ acts as a transposition,
it generates a representation $T$ of the symmetric group~$S_n$ on $\cdn$.
Its action is
\begin{align}
 T({\pi})
 \ket{v_1} \ot \dots \ot
        \ket{v_n}
 =
 \ket{v_{\pi^{-1}(1)}}
        \ot \dots \ot
        \ket{v_{\pi^{-1}(n)}} \,.
\end{align}
We denote the span of $\{T(\pi)\,|\, \pi \in S_n \}$ as $\Sigma_{n}$,
and the subspace of symmetric matrices by $\Sigma_{n}^{\sym} =  \Span\big(\{A \in \Sigma_n \,|\,A^T = A\}\big)$.
We call operators satisfying $(U^{\ot n})^\dag B U^{\ot n} = B$, $\forall U \in \mathcal{U}(d)$
as {\em $U^{\ot n}$-invariant}, or simply {\em unitary invariant}.

\subsection{Schur-Weyl duality}
The Schur-Weyl duality states that the commutant of the algebra $\Sigma_n$ spanned by $\{T(\pi) \,:\, \pi \in S_n\}$
on the space $\cdn$
is the algebra spanned by $A^{\otimes n}$ with $A \in \GL(d)$.
The vector space then decomposes as
\begin{equation}\label{eq:SW}
 \cdn \simeq \bigoplus_{\substack{\lambda \;\vdash \;n \\ \height(\lambda) \leq d}} \UU_\lambda \ot \SS_\lambda\,.
\end{equation}
where the diagonal action $A^{\ot n}$ of the general linear group $\GL(d)$ acts on $\UU_\lambda$
and the symmetric group as $T(\pi)$ on $\SS_\lambda$.
In particular, the QMC Hamiltonian is unitary-invariant, and so is its ground state. 
As a consequence, both can be expanded in terms of permutation operators.

\section{Three qubits}\label{sec:3qb}
We now introduce the key component of our relaxation:
a single three-qubit reduced density matrix, which can be understood as a ``frustration triangle''.
Our aim is to solve this subsystem system exactly,
so that it is consistent with all other three-qubit marginals. 
Henceforth, for simplicity we drop the $T(\pi)$ notation and denote permutations as products of cycles, e.g. $(123)$, with $\id$ denoting the identity permutation.

\subsection{Symmetry-reduction}

In a seminal paper, Eggeling and Werner established a decomposition for tripartite unitary-invariant states~\cite{PhysRevA.63.042111} on $(\C^d)^{\otimes 3}$.
A convenient basis for the invariant space is
\begin{align}
 R_+ &= \frac{1}{6} \big( \id + (12) + (23) + (13) + (123) + (132) \big) \,,\nn\\
 R_- &= \frac{1}{6} \big( \id - (12) - (23) - (13) + (123) + (132) \big) \,,\nn\\
 R_0 &= \frac{1}{3} \big( 2\id - (123) - (132) \big)  \,,\nn\\
 R_1 &= \frac{1}{3} \big( 2(23) - (13) - (12) \big) \,,\nn\\
 R_2 &= \frac{1}{\sqrt{3}} \big((12) - (13) \big) \,,\nn\\
 R_3 &= \frac{i}{\sqrt{3}} \big( (123) - (132) \big) \,.
\end{align}
The projectors $R_+$, $R_-$, and $R_0$ correspond to the three isotypic subspaces labeled by the partitions $[3]$, $[1,1,1]$, and $[2,1]$ respectively.
The operators $R_1, R_2, R_3$ act as Pauli matrices on the subspace corresponding to projector $R_0$. We use the following Lemma:

\begin{lemma}[State characterization~\cite{PhysRevA.63.042111}] \label{lem:EW_pos}
 For any operator $\varrho$ on $\HH^{\ot 3}$ define the six parameters $r_k = \tr(\varrho R_k)$, for $k=\{+,-,0,1,2,3\}$.
 Each unitary-invariant $\varrho$ is uniquely determined by the tuple $(r_+,r_-,r_0,r_1,r_2,r_3) \in \R^6$,
 and such tuple belongs to a density matrix in $\HH^{\ot 3}$ if and only if
 \begin{align}
  r_+, r_-, r_0 \geq 0\,,   \quad\quad r_+ + r_- + r_0 = 1 \,, \quad\quad r_1^2 + r_2^2 + r_3^2 \leq r_0^2\,.
 \end{align}
\end{lemma}
\noindent Identifying $r_0, r_1, r_2, r_3$ with the identity and the Pauli matrices,
the last constraint corresponds to the (subnormalized) Bloch ball of a single qubit.
 The characterization in Lemma~\ref{lem:EW_pos} requires six variables; we now reduce this to three variables:
 First, we show that we may assume that $\varrho$ has real entries,
 and then we show that for real $\varrho$, three variables suffices.

\subsection{Real weights}
Consider a unitary-invariant three-qubit system on particles $1$, $2$, and~$3$. It can be expanded as
\begin{equation}\label{eq:3qb}
 \varrho_{123} = \zeta \one + a {(12)}
                            + b {(23)}
                            + c {(13)}
                            + d {(123)}
                            + e {(132)}\,,
\end{equation}
where $a,b,c,d,e \in \C$.
Now note that the swap operators ${(12)}, {(13)}, {(23)}$ are all Hermitian.
Thus $\expect{{(ij)}} \in \R$.
But also $\expect{{(ijk)}} \in \R$: The QMC Hamiltonian is real, $H^T = H$.
This implies that also its ground state satisfies $\varrho=\varrho^T$.
As a consequence,
\begin{equation}\label{eq:real_3cycles}
    \tr\big(\varrho {(123)} \big)^*
 =  \tr\big(\varrho^\dag {(123)}^\dag \big)
 =  \tr\big(\varrho {(123)}^T \big)
 =  \tr\big(\varrho^T {(123)} \big)
 =  \tr\big(\varrho {(123)} \big)\,,
\end{equation}
where we have used that the permutation operators are orthogonal and that $\pi^T = {\pi^{-1}}$.
From the last equalities also $\expect{{(123)}} = \expect{{(132)}} \in \R$, which we use shortly.

\subsection{Three-qubit identity} \label{sec:simplifications}
On the vector space $(\C^2)^{\ot 3}$, the relation
\begin{equation}\label{eq:3identity}
 \id - {(12)}- {(13)} - {(23)} + {(123)} + {(132)} = 0
\end{equation}
generates all nontrivial identities among the permutation operators~\cite{Procesi2007LieGroups}.
  This simplifies the Werner-Eggeling basis further:
  From the matrix identity \Cref{eq:3identity} one has $r_- = \tr(R_- \varrho) =0$.
  From \Cref{eq:real_3cycles} follows that $\expect{(123)} = \expect{(132)}$, and thus $r_3~=~\tr(R_3 \varrho)~=~0$.
  Via \Cref{eq:3identity}, the remaining operators $R_+, R_0, R_1, R_2$ simplify for QMC to:
  \begin{align}
    R_+ &= \frac{1}{3} \big( (12) + (13) + (23) \big) \,,\nn\\
    R_0 &= \frac{1}{3} \big( 3\id - (12) - (13) - (23) \big)  \,,\nn\\
    R_1 &= \frac{1}{3} \big( 2(23) - (12) - (13) \big) \,,\nn\\
    R_2 &= \frac{1}{\sqrt{3}} \big((12) - (13) \big) \,.
  \end{align}

\subsection{Lieb-Mattis and Parekh-Thompson conditions}
For QMC, Lemma~\ref{lem:EW_pos} simplifies to:

\begin{corollary}\label{cor:sym_red}
 A real unitary-invariant operator in the span of $R_0,R_1,R_2,R_+$ corresponds to a three-qubit state,
 if and only if the following conditions are met:
 \begin{align}\label{eq:conditions_red}
  0 \leq r_+ \leq 1\,, \quad\quad r_+ + r_0 = 1 \,, \quad\quad r_1^2 + r_2^2 \leq r_0^2\,.
 \end{align}
 \end{corollary}

Now write $x = \expect{(12)}$, $y = \expect{(23)}$, and $z = \expect{(23)}$. Assume $\langle \one \rangle = 1$.
Then Corollary~\ref{cor:sym_red} reads:
\begin{align}
   0 \leq x + y + z &\leq 3\,, \label{eq:LM}\tag{LM}\\
   (x^2 + y^2 + z^2) - 2(xy + xz + yz) + 2(x + y + z) &\leq 3\,, \label{eq:PT}\tag{PT}
   \end{align}
The bound ~\eqref{eq:LM} is known as Lieb-Mattis inequality and corresponds to $0\leq r_+ \leq 1$,
and ~\eqref{eq:PT} as the Parekh-Thompson triangle inequality corresponding to $r_1^2 + r_2^2 \leq r_0^2 = (1-r_0)^2$.
The latter was discovered
through an analysis of the second level of the quantum Lasserre hierarchy in Pauli polynomials~\cite[Lemma 7]{parekh2022optimal},
but as we have shown here, it can also be obtained as a consequence of~\cite{PhysRevA.63.042111}.
We have thus shown the following:
\corLMPT*

To see that any single of these constraints are sufficient (for trace one operators),
consider the projectors $R_+$ and $R_0$ onto the subspaces
corresponding to the partitions $[3]$ and $[2,1]$ respectively.
One can check that $\tfrac{1}{2}R_+ - \tfrac{1}{4}R_0$ is a trace one operator satisfying \eqref{eq:PT} but not \eqref{eq:LM},
while $\tfrac{1}{12}\big(R_0 + 2(12)\big)$ is a trace one operator satisfying \eqref{eq:LM} but not \eqref{eq:PT}.

Note that Thm.~\ref{thm:LM_PT} assumes the operator to be of trace one:
there exist non-normalized operators (e.g. $-R_0$) that satisfy both inequalities
but that are not positive semidefinite.
However, if for a given set of $x,y,z\in \R$ \eqref{eq:LM} and \eqref{eq:PT} are satisfied,
then there exists an operator in the hyperplane of trace one that is a state.

A natural question is how these conditions
generalize to the case of $k$-body marginals for dimensions $d > 2$.
As we show in Section~\ref{sec:k-body},
every unitary-invariant $k$-body reduced density matrix can
be decomposed into irreps of the symmetric group $S_k$.
Demanding positivity in each irreducible component with overall trace one yields
necessary and sufficient conditions
for an operator to be a state.
These are generally positive semidefinite cone constraints,
which in the case of real $2\times 2$ matrices reduce to a second order cone.
Further simplifications, that is reducing the number of variables, can be made when $d<k$
by including identities that correspond to partitions
that lie in the kernel of the representation of the symmetric group.
In Section~\ref{sec:simplifications} this is done in the case of three qubits,
for which the representation of the partition $[1,1,1]$ vanishes.

\section{Second order cone relaxation}
\label{sec:soc}

We now use the representation of a real unitary-invariant three-qubit state
to provide a relaxation to QMC
corresponding to optimizing over mutually constituent three-qubit marginals.

\subsection{Compatible marginals}
Consider now an instance of quantum Max Cut on a graph $(G,E)$ with $n$ vertices.
We demand that the approximation to the ground state satisfies:
\begin{enumerate}
\item all three-party subsystems $\varrho_{ijk}$ are density matrices on $(\C^2)^{\ot 3}$;
  \item every pair of three-qubit subsystems are compatible on their joint two-qubit reductions.
\end{enumerate}
We arrive at the following semidefinite programming relaxation to the optimization in Eq.~\eqref{eq:qmc_hamil}:
\begin{align}\label{eq:relax1}
 \underset{\{\varrho_{ijk}\}}{\min} \quad &-\frac{1}{2}\sum_{(i,j) \in E} \tr\big((\one - \swap_{ij}) \varrho_{ij}\big)\nn\\
 \text{subject to}        \quad &\varrho_{ij} = \tr_k (\varrho_{ijk})\,, \nn\\
                                &\varrho_{ijk} \geq 0 \,\nn\\
                                &\tr(\varrho_{ijk} ) = 1\,, \nn\\
                                & \tr_k(\varrho_{ijk}) = \tr_\ell(\varrho_{ij\ell})\,, \quad \quad k \neq \ell\,.
\end{align}
where the sum and constraints are over all pairwise different $i,j,k \in \{1,\dots, n\}$.
This relaxation contains $\binom{n}{3}$ complex SDP variables of size $8\times 8$.
We now show how it can be symmetry reduced to $\binom{n}{2}$ real variables,
with $\binom{n}{3}$ linear programming constraints
and $\binom{n}{3}$ second-order cone constraints.

\subsection{Symmetry reduction}
We now reduce the size of the coupled SDP constraints.
 To every pair $(i,j)$ of vertices associate a variable
 $\E_{ij} \in \R$ representing the expectation value $\langle \swap_{ij} \rangle$.
 If $\E_{ij} = -1$ then $\varrho_{ij}$ is in the singlet state.
 The $\swap$ operator has eigenvalues $\{+1,-1\}$ and thus we have the constraint $-1\leq \E_{ij} \leq 1$.
It follows from \Cref{thm:LM_PT} that
the semidefinite program of \Cref{eq:relax1} can be relaxed to:

\begin{align}\label{eq:relax2}
 \underset{\{\E_{ij}\}}{\min} \quad & - \frac{1}{2} \sum_{(i,j) \in E} \big(1- \E_{ij}\big) \nn\\
  \text{subject to}\quad    &\quad\,0 \leq \E_{ij} + \E_{jk} + \E_{ik} \leq 3\,,\nn\\
                            & \quad\,(\E_{ij}^2 + \E_{ik}^2 + \E_{jk}^2) - 2(\E_{ij}\E_{ik} + \E_{ij}\E_{jk} + \E_{ik}\E_{jk}) + 2(\E_{ij} + \E_{ik} + \E_{jk}) \leq 3
\end{align}
\noindent As desired, this SDP yields
a set of three-qubit density matrices
that are mutually consistent on all two-qubit marginals,
and whose objective value lower bounds the ground state energy of quantum Max Cut.

\subsection{Second order cone relaxation}

A second order cone program (SOCP)  has form~\cite{LOBO1998193}
\begin{align}\label{eq:soc_prog2}
 \underset{x}{\min} \quad & f^T x \nn\\
 \text{subject to}  \quad & \enorm{A_i x + b_i} \leq c_i^T x + d_i\,, \quad i=1,\dots, N\,,
\end{align}
where $\enorm{\cdot}$ is the Euclidean norm.
The variable $x\in\R^n$,
and the problem parameters are
$f \in \R^n$,
$A_i \in \R^{n_i-1 \times n}$,
$b_i \in \R^{n_i-1}$,
$c_i \in \R^n$, and
$d_i \in \R$.
Equality constraints of the form $A_i x = -b_i$ can be included by setting $c_i = d_i = 0$ in~\eqref{eq:soc_prog2}.

 Note that Eq.~\eqref{eq:relax2} [respectively
 the last constraint in Eq.~\eqref{eq:conditions_red}] is a second order cone constraint,
 \begin{equation}
  ||A_{ijk} \vec \E_{ijk} + b_{ijk} || \leq c_{ijk} \vec \E_{ijk} + d_{ijk}\,, \quad 1\leq i < j < k \leq n\,.
 \end{equation}
 with
 \begin{equation}
  \vec \E_{ijk} = \begin{pmatrix}
                  \E_{ij} \\
                  \E_{ik} \\
                  \E_{jk}
                 \end{pmatrix}\,,
  \quad
  A_{ijk} = \frac{1}{3}
  \begin{pmatrix}
   -1  & -1 & 2 \\
   \sqrt{3}  & -\sqrt{3}  & 0
  \end{pmatrix}\,,
  \quad
  c_{ijk} =  \frac{1}{3}
  \begin{pmatrix}
   -1 \\ -1 \\ -1
  \end{pmatrix}\,,
  \quad
  b_{ijk} = d_{ijk} = 0\,.
  \end{equation}

This gives the following result:
\thmSOC*

\subsection{Why SOCP}
On a conceptual level, why does a second order cone constraint
appear in our relaxation?
The reason is that the positivity condition of a single three-qubit marginal
decomposes into a linear constraint on the $[3]$ partition,
and a positive semidefinite constraint on
the $2$-dimensional $[2,1]$ partition,
\begin{equation}
 \begin{pmatrix}
  r_0+r_1 & r_2 \\
  r_2 & r_0 - r_1
 \end{pmatrix}
 \succeq 0\,.
\end{equation}
By the determinant test is is clear that this constraint is equivalent to $r_1^2 + r_2^2 \leq r_0^2$ with $r_0-r_1 \geq 0$
and $r_0 + r_1 \geq 0$.

\section{Approximation algorithm}
\label{sec:approx}

We will find it convenient to cast QMC as maximization problem, by negating the terms in  \eqref{eq:qmc_hamil}, so that energies are nonnegative.

\subsection{Relaxation}

The quantum Lasserre hierarchy indexing by polynomials in Pauli matrices can be used to strengthen the SOCP relaxation.
Consider three moment matrices with entries
\begin{align}
 M^X_{ij} &= \tr(X_i X_j \varrho) \,, &
 M^Y_{ij} &= \tr(Y_i Y_j \varrho) \,, &
 M^Z_{ij} &= \tr(Z_i Z_j \varrho) \,,
\end{align}
where $X_i$, $Y_i$, $Z_i$ are single Pauli operators on position $i$.
It is clear that $M^X, M^Y, M^Z \succeq 0$. Then consider their sum
$ M = M^X + M^Y + M^Z$,
which by \Cref{eqn:SWAP} has the properties
\begin{align}
 M\succeq 0\,,\quad M_{ii} = 3\,,\quad M_{ij} = 2\E_{ij} - 1\,.
\end{align}

With this additional semidefinite constraint, the relaxation~\eqref{eq:relax2} is strenghtened to:

\begin{align}\label{eq:Pauli1}
 \underset{\E_{ij}}{\max} \quad & \frac{1}{2} \sum_{(i,j) \in E} \big(1- \E_{ij}\big) \nn\\
 \text{subject to}\quad     &\quad\,0 \leq \E_{ij} + \E_{jk} + \E_{ik} \leq 3\,,\nn\\
                            & \quad\,(\E_{ij}^2 + \E_{ik}^2 + \E_{jk}^2) - 2(\E_{ij}\E_{ik} + \E_{ij}\E_{jk} + \E_{ik}\E_{jk}) + 2(\E_{ij} + \E_{ik} + \E_{jk}) \leq 3\nn\\
 & \quad\, M\succeq 0 \,, \quad
 M_{ii} = 3 \,,\quad
 M_{ij} = 2\E_{ij} -1\,.
\end{align}

\subsection{Algorithm and approximation guarantee}

The recent design of approximation algorithms for Quantum Max Cut has exposed new connections between convex programming relaxations of local Hamiltonians and monogamy of entanglement.
In particular, the second level of the quantum Lasserre hierarchy in Pauli polynomials enforces the Lieb-Mattis inequality on stars in the interaction graph
(this inequality for stars with two edges follows from \eqref{eq:LM})~\cite{https://doi.org/10.4230/lipics.icalp.2021.102}.
This star bound has been instrumental in designing approximation algorithms for Quantum Max Cut,
and seems to be a necessary ingredient in obtaining good approximation algorithms using entangled states.
We will work with with variables $y_{ij} = \frac{1}{2}(1-x_{ij}) = \tr(\frac{1}{2}(\one - \swap_{ij}) \varrho)$, for a state $\varrho$, corresponding to the negative of Quantum Max Cut Hamiltonian terms \eqref{eq:qmc_hamil}.
The $y_{ij}$ variables will also correspond to relaxed values in the context of our SOCP and SDP relaxations.
The analog of equations \eqref{eq:LM} and \eqref{eq:PT} in $y_{ij}$ variables are

\begin{align}
   &0 \leq y_{ij}+y_{jk}+y_{ik} \leq \frac{3}{2} \label{eq:LM'}\tag{LM'}\\
   &y_{ij}^2+y_{jk}^2+y_{ik}^2 \leq 2(y_{ij} y_{jk}+y_{ij} y_{ik}+y_{jk}y_{ik})\label{eq:PT'}\tag{PT'}
   \end{align}

\begin{lemma}[Star bound, Theorem 11 in~\cite{https://doi.org/10.4230/lipics.icalp.2021.102}] Let $y_{ij} \in [0,1]$ for all $(i,j) \in E$ correspond to values on Quantum Max Cut Hamiltonian terms from the second level of the quantum Lasserre hierarchy in Pauli polynomials.
Then for each $i \in V$,
 \begin{equation}
 \sum_{j \in N(i)} y_{ij} \leq \frac{|N(i)| + 1}{2},
 \end{equation}
where $N(i) = \{j \,:\, (i,j) \in E \}$ denotes the neighborhood of $i$.  Equivalently,
\begin{equation}\label{eq:star_bound_matching}
-\sum_{j \in N(i)} x_{ij} \leq 1.
\end{equation}
\end{lemma}

Approximation algorithms for Quantum Max Cut hinge on a selection of an appropriate ansatz for output states.  For product states, an $\alpha$-approximation algorithm for Quantum Max Cut, with $\alpha \approx 0.498$, is best possible~\cite{hwang2023unique}, and an $\alpha$-approximation is possible using the first level of the Pauli-based quantum Lasserre hierarchy~\cite{gharibian_et_al:LIPIcs:2019:11246}.  An optimal $\frac{1}{2}$-approximation using product states is possible by employing the star bound and the Parekh-Thompson triangle inequality~\cite{parekh2022optimal}.  Anshu, Gosset, and Morenz~\cite{https://doi.org/10.4230/lipics.tqc.2020.7} observed that the $\frac{1}{2}$-approximation barrier could be surmounted by using tensor products of 1- and 2-qubit states, which can be produced by finding matchings\footnote{A matching of $G$ is a subset of edges such that no two edges in the subset share a common vertex.} in the interaction graph $G$.  Given a matching $M$, placing singlets on edges $e \in M$ yields a state.  The currently best known 0.595-approximation by Lee and Parekh is obtained by simply taking the better of the state obtained from a maximum matching in $G$ or the product state obtained by the Gharibian-Parekh $\alpha$-approximation on a solution from the second level of the Pauli hierarchy~\cite{lee2024improved}.

The star bound, particularly in the form of \eqref{eq:star_bound_matching}, illustrates connections to matching, where for a matching $M$, we must have, for $i \in V$, $\sum_{j \in N(i)} m_{ij} \leq 1$  if $m_{ij} \in \{0,1\}$ are variables indicating $ij \in M$.  Our SDP or SOCP relaxations do not imply the star bound, and in absence of it, proving approximation guarantees beyond 0.498 is challenging using existing techniques~\cite{gharibian_et_al:LIPIcs:2019:11246,https://doi.org/10.4230/lipics.icalp.2021.102,parekh2022optimal,lee2022optimizing,King2023improved,lee2024improved}.  As with previous approaches, our algorithm outputs a product state as well as a state derived from a matching.  However, in addition we introduce a new technique that allows us to get a 0.526-approximation with respect to our SDP relaxation \eqref{eq:Pauli1}.  Previous approaches have used SDP values on edges to decide how to round to a quantum state.  While these values depend on the input graph $G$, the relationship is not direct.  Ideally one would be able to show that a set of SDP values resulting in a bad approximation guarantee cannot occur as an optimal SDP solution for an actual input graph.  Our analysis attempts to mimic this by keeping track of the fraction of edges with low, medium, or high SDP value and allows these fractions to be chosen adversarially.


\begin{algorithm}
\caption{ SOCP Approximation Algorithm}\label{alg:approx}
\begin{algorithmic}[1]
\Require Graph $G=(V, E)$, edge weights $\{w_{ij}\}_{ij\in E}$, threshold $t>3/4$
\State Solve \eqref{eq:Pauli1} using graph $G$ and edge weights $\{w_{ij}\}_{ij\in E}$ to obtain optimal variable values $(M, \{y_{ij}\}_{ij\in E}\})$.
\State Construct $S=\{ij \in E: y_{ij} > t \}$ (guaranteed to be a matching by \Cref{eq:star_bound_matching} since $y_{ij} > 3/4$ implies $-x_{ij} > 1/2$)
\State Compute Cholesky decomposition of $M$ to obtain vectors $\{\ket{v_i}\in \mathbb{R}^{|V|}\}$ such that $M_{ij}=\langle v_i|v_j\rangle$.
\State Sample $R\in \mathbb{R}^{3 \times |V|}$ by sampling matrix elements $R_{ij}$ as i.i.d. standard normal random variables.
\State For each $k\in V$ define $(\theta_X^k, \theta_Y^k, \theta_Z^k)= R \ket{v_k}/||R\ket{v_k}||_2$.
\State Define $\varrho_S=\prod_{ij\in S } \left( \frac{\mathbb{I}-X_i X_j -Y_i Y_j-Z_i Z_j}{4} \right)\prod_{k:ik\notin S \,\,\forall i}\left( \frac{\mathbb{I}+\theta_X^k X_k+\theta_Y^k Y_k+\theta_Z^k Z_k}{2}\right)$
\State Define $\varrho_{prod}=\prod_{k\in V} \left( \frac{\mathbb{I}+\theta_X^k X_k+\theta_Y^k Y_k+\theta_Z^k Z_k}{2}\right)$.
\State \Return $\max\{\tr(H\varrho_S), \tr(H \varrho_{prod})\}$
\end{algorithmic}
\end{algorithm}

We obtain the following approximation result.

\begin{theorem}\label{thm:approx}
For $t=0.771$, \cref{alg:approx} is a 0.526  approximation algorithm for Quantum Max Cut.
\end{theorem}

\onote{The $R_{ij}$ are just random Guassians and remain the same, but I did change everything to lowercase $y_{ij}$ for consistency.}

\begin{lemma}\label{lem:tech_rounding}
Let $\{y_e\}_{e\in E}$ and $M$ be derived from a feasible solution to \eqref{eq:Pauli1}.  Given a feasible solution let $\{(\theta_X^k, \theta_Y^k, \theta_Z^k)\}_{k\in V}$ be sampled as in \Cref{alg:approx} and define
$$
F(x):= \frac{1}{4x}\left(1-\frac{8}{9 \pi} \,_2 F_1\left(\frac{1}{2}, \frac{1}{2};\frac{5}{2};\frac{1}{9}(1-4 x)^2\right)\right),
$$ 
where $\,_2F_1$ is the Hypergeometric function.
\begin{enumerate}
\item\label{lem:tech_rounding_item1}~\cite{gharibian_et_al:LIPIcs:2019:11246, v010a004} $$
\mathbb{E}_R \,\, \left(\frac{1-\theta_X^i \theta_X^j-\theta_Y^i \theta_Y^j-\theta_Z^i \theta_Z^j}{4 Y_{ij}}\right)=F\left(\frac{1-M_{ij}}{4}\right).
$$
\item\label{lem:tech_rounding_item2} ~\cite{gharibian_et_al:LIPIcs:2019:11246} $F(x) \geq 0.498$ for $x\in (0, 1]$.
\item\label{lem:tech_rounding_item3}  $F(x)$ is monotonically decreasing in the interval $(0, 4/5]$.
\item\label{lem:tech_rounding_item4} ~\cite{lee2024improved} If $y_{ij} \geq 3/4$ then $y_{jk} \leq 1/4 (3-2 y_{ij}+2 \sqrt{3} \sqrt{y_{ij}-y_{ij}^2})$.
\end{enumerate}
\end{lemma}
\begin{proof}
Only \Cref{lem:tech_rounding_item3} requires a proof\footnote{Technical lemmas very similar to \Cref{lem:tech_rounding_item3} were used in other works~\cite{hwang2023unique, parekh2022optimal, https://doi.org/10.4230/lipics.icalp.2021.102}, but we could not find anything which explicitly implies \Cref{lem:tech_rounding_item3}.}.  Let $q$ be shorthand for $\frac{1}{3}(1-4x)$.  Using the known derivative $\frac{d}{dz}\,_2 F_1(a, b;c; z)=\frac{ab}{c} \,_2 F_1(a+1, b+1;c+1; z)$~\cite{abramowitz1948handbook} we may compute
\begin{align}
\frac{dF}{dx}=\frac{32 x q^2 \, _2F_1\left(\frac{3}{2},\frac{3}{2};\frac{7}{2};q^2\right)+40 \, _2F_1\left(\frac{1}{2},\frac{1}{2};\frac{5}{2};q^2\right)-45 \pi }{180 \pi  x^2} =: \frac{g(x)}{180 \pi  x^2}.
\end{align}
We can see that $\frac{dF}{dx} \leq 0$ for all $x\in (0, 4/5]$ if and only if $g(x) \leq 0$ for all $x\in (0, 4/5]$.  We may compute
\begin{align}
\frac{dg}{dx}=-\frac{384}{7} x (1-4x) \left(3 q^2 \, _2F_1\left(\frac{5}{2},\frac{5}{2};\frac{9}{2};q^2\right)+7 \, _2F_1\left(\frac{3}{2},\frac{3}{2};\frac{7}{2};q^2\right)\right).
\end{align}
By the series expansion $\,_2 F_1 (z) =\sum_{n=0}^\infty \frac{(a)_n (b)_n z^n}{(c)_n n!} $, one can see that $ (3 q^2 \, _2F_1(\frac{5}{2},\frac{5}{2};\frac{9}{2};q^2)+$ $7 \, _2F_1 (\frac{3}{2},\frac{3}{2};\frac{7}{2};q^2))$ is non-negative where defined since $\,_2 F_1 (z^2)$ is a series of non negative terms.  
Note $x\leq 4/5$ by assumption so $\,_2 F_1(5/2, 5/2;9/2, 1/9(1-4x)^2)$ is always defined.  Hence $dg/dx \leq 0$ for $x \leq 1/4$ and $dg/dx \geq 0$ for $1/4 \leq x \leq 4/5$.  It follows that $g(x) \leq \max\{g(0), g(0.8)\}$ for $x \in [0, 4/5]$.  $g(0)$ and $g(4/5)$ can be computed to show $g(x) \leq 0$ for all $x\in [0, 4/5]$ which implies $dF/dx \leq 0$ for all $x\in (0, 4/5]$.  This implies $F$ is monotonically decreasing in this interval.
\end{proof}

\begin{proof}[Proof of \Cref{thm:approx}]
For each edge $e$ let $G_e^S$ be the expected value of edge $e$ obtained by the rounding algorithm for the state $\varrho_S$ and let $G_e^{prod}$ be defined analogously.  Let $S$ be the set of edges with $y_e >t$, $T$ be the set of edges adjacent to $S$, and $U$ be the set of edges not included in and not adjacent to $S$.  Since $t=0.771$, by \Cref{lem:tech_rounding} \Cref{lem:tech_rounding_item4}, edges $e\in T$ have values $y_e \leq t' =1/4 (3-2 t+2 \sqrt{3} \sqrt{t-t^2})\approx 0.728$.

First consider the approximation ratio for $\varrho_{prod}$.
\begin{align}
\frac{\sum_e w_e G_e^{prod}}{\sum_e w_e y_e}=\sum_e \left(\frac{w_e y_e}{\sum_{e'} w_{e'} y_{e'}} \right) \frac{G_e^{prod}}{y_e}=\sum_e \alpha_e \frac{G_e^{prod}}{y_e}\\
\nonumber =\sum_{e\in S} \alpha_e \frac{G_e^{prod}}{y_e}+\sum_{e\in T} \alpha_e \frac{G_e^{prod}}{y_e}+\sum_{e\in U} \alpha_e \frac{G_e^{prod}}{y_e}.
\end{align}
By \Cref{lem:tech_rounding} \Cref{lem:tech_rounding_item1} $G_e^{prod}/y_e=F(y_e)$.  Since $F(x)$ is monotonically decreasing by \Cref{lem:tech_rounding} \Cref{lem:tech_rounding_item3} we may lower bound $G_e^{prod}/y_e \geq F(t)$ for $e\in U$ and $G_e^{prod}/y_e \geq F(t')$ for $e\in T$.  By \Cref{lem:tech_rounding} \Cref{lem:tech_rounding_item2} we may lower bound $G_e^{prod}/y_e \geq 0.498$ for $e\in S$.  Hence,
\begin{align}\label{eq:low1}
\frac{\sum_e w_e G_e^{prod}}{\sum_e w_e y_e} \geq \alpha 0.498 + \beta F(t') +\gamma F(t),
\end{align}
for $\alpha=\sum_{e\in S} \alpha_e$, $\beta=\sum_{e\in T} \alpha_e$ and $\gamma=\sum_{e\in U} \alpha_e$.

Now consider $\varrho_S$.  $G_e^S/y_e \geq 1$ for $e\in S$ since $\varrho_S$ is a singlet along that edge, $G_e^S/y_e \geq 1/4$ for $e \in T$ since one endpoint of edge $e$ is involved in a singlet state and $G_e^S/y_e \geq F(t)$ for $e\in U$ by \Cref{lem:tech_rounding} \Cref{lem:tech_rounding_item1}.  Similar reasoning yields
\begin{align}\label{eq:low2}
\frac{\sum_e w_e G_e^{S}}{\sum_e w_e y_e} \geq \alpha +\frac{\beta}{4}+\gamma F(t).
\end{align}

If we denote $G_e$ as the expected edge value from the overall algorithm then 
\begin{align}
\frac{\sum_e w_e G_e}{\sum_e w_e y_e} \geq \frac{\max\{\sum_e w_e G_e^{prod}, \sum_e w_e G_e^S \}}{\sum_e w_e y_e}\geq \max\left\{ \frac{\sum_e w_e G_e^{prod}}{\sum_e w_e y_e}, \frac{\sum_e w_e G_e^S}{\sum_e w_e y_e} \right\}.
\end{align}
By \eqref{eq:low1} and \eqref{eq:low2},
\begin{align}
\frac{\sum_e w_e G_e}{\sum_e w_e y_e} \geq \max\{0.498 \alpha+\beta F(t')+\gamma F(t), \alpha +\frac{\beta}{4}+\gamma F(t)\}.
\end{align}
Since $\sum_e \alpha_e=1$ and each $\alpha_e \geq 0$ we can lower bound the approximation factor for the overall algorithm as 
\begin{align}
r:=\min_{\substack{\alpha+\beta+\gamma=1\\ \alpha, \beta, \gamma \geq 0}} \,\, \max\{0.498 \alpha+\beta F(t')+\gamma F(t), \alpha +\frac{\beta}{4}+\gamma F(t)\}.
\end{align}
$r$ may be solved with a linear program
\begin{align}
r=& \min_{\alpha, \beta, \gamma, s} s&\\
\nonumber \text{subject to}\,\,\,\,\,\,& 0.498 \alpha+\beta F(t')+\gamma F(t) \leq s\\
\nonumber &\alpha +\frac{\beta}{4}+\gamma F(t) \leq s\\
\nonumber &\alpha +\beta +\gamma =1\\
\nonumber &\alpha, \beta, \gamma \geq 0.
\end{align}
We obtain $r=0.526$.

\end{proof}

\section{k-body relaxations}
\label{sec:kbody}

\label{sec:k-body}
We generalize the three-qubit state constraints to $k$-body reduced density matrices with $d\geq 2$.
Theorem~\ref{thm:pos} characterizes positivity of an invariant operator directly from its expectation values
and might be of independent interest.

\subsection{Weingarten expansion}
One can expand a unitary-invariant operator as~\cite{CollinsSniady2006, procesi2020note}
\begin{equation}\label{eq:Weingarten_expansion}
 A = \sum_{\sigma \in S_k} \tr(\sigma^{-1} A) \sigma \Wg(d,n)\,,
\end{equation}
where the Weingarten {\em operator} is given by
\begin{align}
 \Wg(d,k) &= \frac{1}{k!} \sum_{\substack{\lambda \vdash k \\ |\lambda| \leq d}} \frac{\chi_\lambda(\id)}{s_{\lambda,d}(1)} p_\lambda\,.
\end{align}
Here $\chi_\lambda$ is the character of the symmetric group,
$s_{\lambda,d} = s_\lambda(1,\dots, 1)$ ($d$ times) the Schur polynomial,
and the central Young projectors defined by
\begin{equation}\label{eq:young}
 p_\lambda = \frac{\chi_\lambda(\id)}{k!} \sum_{\pi \in S_k} \chi_\lambda(\pi^{-1}) \pi\,.
\end{equation}
we can write the Weingarten {operator} as
\begin{equation}
 \Wg(k,d) = \frac{1}{(k!)^2}
 \sum_{\substack{\lambda \vdash k \\ |\lambda| \leq d}}
 \frac{\chi_\lambda(\id)^2}{s_{\lambda,d}(1)}
 \sum_{\pi_\in S_k} \chi_\lambda(\pi^{-1}) \pi \,.
\end{equation}

With this we can expand Eq.~\eqref{eq:Weingarten_expansion} as
\begin{align}\label{eq:exp_value_expansion}
 A &= \sum_{\sigma, \tau \in S_k} \tr(\sigma^{-1} A) \tau \Wg(\tau^{-1} \sigma, d)
\end{align}
where, using the fact that $\chi_\lambda(\pi) = \chi_\lambda(\pi^{-1})$, the Weingarten {\em function} is
\begin{equation}\label{eq:weingarten_elementwise_expansion}
 \Wg(\pi, d) =
 \frac{1}{(k!)^2}
 \sum_{\substack{\lambda \vdash k \\ |\lambda| \leq d}}
 \frac{\chi_\lambda(\id)^2}{s_{\lambda,d}(1)}
 \chi_\lambda(\pi)\,.
\end{equation}
Expansion ~\eqref{eq:exp_value_expansion} has the drawback that it requires to compute the Weingarten operator respectively function.
We address this in Section~\ref{sec:direct}, Theorem~\ref{thm:pos}.

\subsection{Symmetry reduction}

Let $A \in \Sigma_{k,d}$ and $R_\lambda$ be an orthogonal representation corresponding to $\lambda$.
By the Schur-Weyl decomposition in Eq.~\eqref{eq:SW},
\begin{equation}
 T \mapsto \bigoplus_{\lambda \vdash k\,,  |\lambda| \leq d} R_\lambda
\end{equation}
is a $*$-isomorphism. These preserves positivity~\cite{Bachoc2012}, and thus
\begin{equation}\label{eq:psd_decomp}
A = \sum_{\pi \in S_k} a_\pi T(\pi) \succeq 0
\quad \iff \quad
\bigoplus_{\substack{\lambda \vdash k \\ |\lambda| \leq d}}
 \sum_{\pi \in S_k} a_\pi R_\lambda(\pi) \succeq 0\,.
\end{equation}
In terms of the expectation values $\tr(\pi A)$,
the coefficients $a_\pi$ are given in terms of expectation values $\langle \sigma^{-1}\rangle = \tr(\sigma^{-1} A)$ as
\begin{equation}\label{eq:a_coeffs}
 a_\pi = \sum_{\sigma \in S_k} \tr(\sigma^{-1} A) \Wg(\pi^{-1} \sigma, d)\,.
\end{equation}

\subsection{Direct decomposition}
\label{sec:direct}
Eq.~\eqref{eq:a_coeffs} suggests that Weingarten function has to be evaluated to test for positivity given expectation values.
We now show how a more direct decomposition can be made which avoids this.
The key idea is that the Weingarten operators scales each isotypic component
$\lambda$ by a positive scalar, and that $R_\mu(p_\lambda) = \delta_{\mu\lambda} \one$.
Now use the fact that $p_\lambda$ is in the center of the group ring $\C[S_n]$,
 \begin{align}
 A &= \sum_{\sigma \in S_k} \tr(\sigma^{-1} A) \sigma \Wg(d,n) \nn\\
   &= \frac{1}{k!} \sum_{\substack{\lambda \vdash k \\ |\lambda| \leq d}} \frac{\chi_\lambda(\id)}{s_{\lambda,d}(1)} p_\lambda
   \sum_{\sigma \in S_k} \tr(\sigma^{-1} A) \sigma
\end{align}
Then
\begin{equation}\label{eq:pos_in_irrep}
 A \succeq 0 \quad \iff \quad
  \bigoplus_{\substack{\mu \vdash k \\ |\mu| \leq d}}
  \frac{1}{k!} \sum_{\substack{\lambda \vdash k \\ |\lambda| \leq d}} \frac{\chi_\lambda(\id)}{s_{\lambda,d}(1)}
  R_\mu\Big(p_\lambda \sum_{\sigma \in S_k} \tr(\sigma^{-1} A) \sigma \Big) \succeq 0\,.
\end{equation}
We simplify this expression in two steps:
First note that the centrally primitive idempotents project onto the isotypic components labeled by $\lambda$
and thus $R_\lambda(p_{\mu}) = \one \delta_{\lambda\mu}$.
This can be seen by the fact that $p_\lambda$ projects onto a minimal left ideal of $\C S_n$,
while $R_\mu$ is an irreducible representation of $S_n$.
Under a representation, left ideals are mapped to invariant subspaces;
and when the ideal is minimal then its representation is either irreducible or zero.
Thus when $\lambda = \mu$, then
$R_\lambda(p_\lambda)$ acts as the identity onto the subspace $\lambda$,
and when $\lambda \neq \mu$ then $R_\mu(p_\lambda) = 0$.
Consequentially, $ R_\lambda(p_{\mu}) = \one \delta_{\lambda\mu}$.
Second, when $|\lambda| \leq d$ then $\frac{\chi_\lambda(\id)}{s_{\lambda,d}(1)} > 0$ is a positive number.

As a consequence, we have the following simplification of Eq.~\eqref{eq:psd_decomp} respectively Eq.~\eqref{eq:pos_in_irrep}:
\thmPOS*

\subsection{Variable reduction}

The ground state of the Quantum Max Cut Hamiltonian has real coefficients $a_\pi = a_\pi^{-1} \in \R$.
To see this, note that
$H^T = H$. The ground state can thus be chosen to satisfy the same symmetry. Consequently,
$\varrho^T = \varrho$. Then
\begin{equation}
    \tr\big(\varrho \pi \big)^*
 =  \tr\big(\varrho^\dag \pi^\dag \big)
 =  \tr\big(\varrho \pi^T \big)
 =  \tr\big(\varrho^T \pi \big)
 =  \tr\big(\varrho \pi \big)\,.
\end{equation}
Thus the expectation values are real, $\langle \pi \rangle \in \R$.
In Eq.~\eqref{eq:psd_decomp} the expectation values are multiplied by the Weingarten function in Eq.~\eqref{eq:weingarten_elementwise_expansion} which is also real.
Thus the coefficients satisfy $a_\pi = a_{\pi^{-1}} \in \R$ for all $\pi \in S_k$.

A more systematic variable reduction can be done through the identities generated
by the projector $R_{-}$ onto the anti-symmetric subspace of three systems.
An alternative is to use the {\em Good Permutations} basis~\cite{procesi2020note}
works for general $d$ but is overcomplete in the case of symmetric matrices.

\section{Numerical results}
\label{sec:numerics}
In this section we detail some numerical results on the performance of our relaxations.  
We primarily investigated (1) the performance of our graph relative to known variation benchmarks~\cite{wu2023variational} (VarBench) (2) the performance/behavior on the Shastry-Sutherland model ~\cite{shastry1981exact} and (3) the performance on random (Erdos-Renyi) graphs.

In the first set of benchmarks we compared our relaxations to known ground states and heuristically derived solutions for a subset of lattices covered in~\cite{wu2023variational}.  
For these results the QMC Hamiltonian is scaled to be traceless, i.e. $H=\sum_{ij \in E} X_i X_j +Y_i Y_j +Z_i Z_j$.
Our calculations were all completed on a mid-range workstation without supercomputing resources in contrast to~\cite{wu2023variational}.  
In particular with a Intel(R) Xeon(R) Gold 6258R processor and 64 GB of RAM.  
We stress additionally that our methods are complementary to~\cite{wu2023variational} since we provide lower rather than upper bounds on the ground state energy (which is normally achieved with heuristics).  
One peculiar feature of the optimal solutions we calculated was that in many cases the optimal solution (in the VarBench scaling) is exactly minus the number of edges.  
In the context of \Cref{eq:relax2} this corresponds to setting all $\E_{ij}$ variables to $1$.  This will always be a feasible solution to these programs so the optimal solution of the relaxation will always be smaller than the number of edges.  We compared the relaxation strengths and computation times of SOC and SOC+P1 to the level-1 SWAP program defined in \cite{takahashi2023su2symmetric, Watts2024relaxationsexact}.  Since our goal was primarily to analyze the relaxations we presented in our paper we set an upper bound of 600 seconds for the level-1 SWAP program.  For examples in which we achieved convergence, the level-1 SWAP program was much slower, often by orders of magnitude.  The level-1 SWAP program is an essential ingredient for best known rigorous approximations \cite{lee2024improved} to QMC so the table indicates that our rounding algorithm \Cref{alg:approx} provides practical runtimes for much larger Hamiltonians than previous work.

We also investigated the Shastry-Sutherland model~\cite{shastry1981exact}.  
This is a $2d$ model with couplings $J$ along the edges of a $2D$ grid and couplings $J_D$ across the diagonal of each face in the grid (see \cite{wietek2019thermodynamic} Eq. $1$ for a definition).  
This model is known to have two phase transitions where the quantum state rapidly changes from one form to another.  
The first is at $J/J_D\approx 0.65$ and the second at $J/J_D \approx 0.75$.  
For $J/J_D\leq 0.65$ the ground state is known to be a tensor product of singlets across the diagonal edges of the lattice, while for $J/J_D \geq 0.75$ the state is known to be in the Néel phase corresponding to the product state which maximizes the terms on the square lattice.
In between the ground state is a transient plaquette state.  
Our relaxations mimicked this behavior (see \Cref{fig:ShastrySutherland}) although our relaxations underwent a ``phase transition'' at different parameter values and the transient regime was much shorter.  We can see evidence of this in \Cref{fig:SS_opt} where we plot the optimal values of different relaxations as well as the ground state energy for a Shastry Sutherland model on $16$ qubits.  Indeed we can notice a cusp in the objective of our relaxation at $J/J_D=0.5$ while the exact energy has a cusp at $J/J_D\approx 0.65$.  We note that below $J/J_D=0.5$ all the relaxations considered are correct.  The ``phase transition'' in our relaxations also apparently corresponds to a point at which the relaxations become inexact similar to what was observed in \cite{takahashi2023su2symmetric}.



Additionally we added disorder to the Shastry-Sutherland model and investigated the objective as a function of disorder.  
For these models we started at parameter setting where $J/J_D=0.4$ and multiplied each coupling term by $(1+\sigma X)$ where $\sigma$ was a parameter we varied and $X$ was a standard normal random variable.  
The numerics indicate that the relaxations we have defined are exact for small amounts of disorder and have good performance for even relatively large amounts of disorder. 
Application of convex optimization hierarchies to studying perturbation theory is a very new and interesting research direction~\cite{hastings2022perturbation} and we provide evidence that even the very simple relaxations we have defined are powerful tools for understanding the change in ground state energy under a small perturbation.  

We also studied the performance of our relaxations on randomly generated Erdos-Renyi graphs with probability $p$ of edge formation for several values of $p$ in \Cref{fig:ER_objectives} and \Cref{fig:er_ratios}.  
We believe these are important instances to consider for applications of our ideas because heuristics which make use of the structure of the Hamiltonian (i.e. those considered in~\cite{wu2023variational}) will fail on these instances.  
For dense instances it appears that the SOC program plus level-$1$ Pauli (denoted $SOC+P1$) has appreciably better performance than just the $SOC$ program \Cref{eq:relax2} and that the level-1 SWAP program has appreciably better performance than the other two.  
This is seen in both \Cref{fig:ER_objectives} and \Cref{fig:er_ratios} (note that the Hamiltonian in \Cref{fig:ER_objectives} is an affine shift of \Cref{fig:er_ratios}).  
For all the instances we considered we saw reasonable performance with our relaxations indicating that the relaxations are getting within $25\%$ of the optimal value, but we cannot probe very big examples due to the exponential dimension.

\begin{table}
 \begin{tabular}{@{} l r r r r r r r@{}}
 \toprule
 model        &  \# edges & SOC    & \thead{SOC\\+Pauli-1} & \thead{Level-1\\SWAP} & VarBench & \thead{Rounded\\Objective} &\thead{Comp. Time SOC/SOC+\\Pauli-1/Level-1 SWAP (s)}  \\
 \midrule
 square   16      &  32       & -64     & -53.3333    & -45.57   & -44.9139$^*$        &-15.67      & 0.05/1.5/1.25      \\
 square   36      &  72       & -144    & -123.4306   & -99.76   & -97.7576$^*$        & -36.01     & 1.4/2.5/256  \\
 square   64      &  128      & -256    & -220.1495   &          &                     & -64.86     & 19/54         \\
 square   50      &  100      & -200    & -172.405    &          & -135.0204$^*$       & -51.04     &6.7/27        \\
 square  100      &  200      & -400    &  -344.950   &          & -268.560\phantom{0} & -101.87    &258/1145      \\
 square  196      &  392      & -784    & -676.932    &          & -523.983            & -199       &21357/20024        \\
 kagome   18      &  36       &  -36    & -36         & -36      &  -32.1931$^*$       & -10.34     &0.08/0.08/1.52     \\
 kagome   48      &  96       &  -96    & -96         &          &  -84.2311$^*$       & -27.36     &2.2/2.7       \\
 kagome   75      &  150      &  -150   & -150        &          &                     & -43.22     &11/16         \\
 kagome  192      &  384      &  -384   & -384        &          & -330.143\phantom{0} & -110.43    &1166/2274     \\
 shuriken 24      &  48       &  -48    & -48         & -48      &  -43.0396$^*$       & -13.72     &0.18/0.28/7.39     \\
 shuriken 96      &  192      &  -192   & -192        &          & -168.2918           & -55        &36/44         \\
 Pyrochlore 32    &   96      & -120    & -109.838    & -82.62   & -66.15$^*$          & -31.77     &2.1/3.2/41       \\
 Pyrochlore 108   &   324     & -405    &  -393.133   &          &  -210.43            & -114.12    &592/646       \\
 Pyrochlore 256   &  768     &  -960    &  -933.512   &          &  -494.69            & -271.22    &25835/31584   \\
 \bottomrule
 \end{tabular}
 \caption{
  {\bf Bounds on QMC} for periodic lattices using the
  scaled Hamiltonian $H = \sum_{(ij) \in E} (XX + YY + ZZ)_{ij}$.
  Lower bounds include the second order cone (SOC), second order cone plus Pauli Level 1 (SOC+Pauli-1) as well as the Level 1 SWAP hierarchy Rounded objective corresponds to the upper bound obtained through \Cref{alg:approx}.
  In the VarBench column the symbol $^*$ indicates that the exact solution was found while the other entries generally correspond to heuristic upper bounds on the ground state energy.  We note that square lattice of size $50$ is defined differently than the other square lattices, see~\cite{PhysRevE.98.033309} Fig.~$5$ for a definition.
  Blank entries under Level-1 SWAP column correspond to lattices which exhausted the computers memory or took longer than the set threshold of 600 seconds to converge. }

\end{table}

\begin{figure}[tbp]
\begin{subfigure}{0.49\textwidth}
\centering
 \includegraphics[width= 1\textwidth]{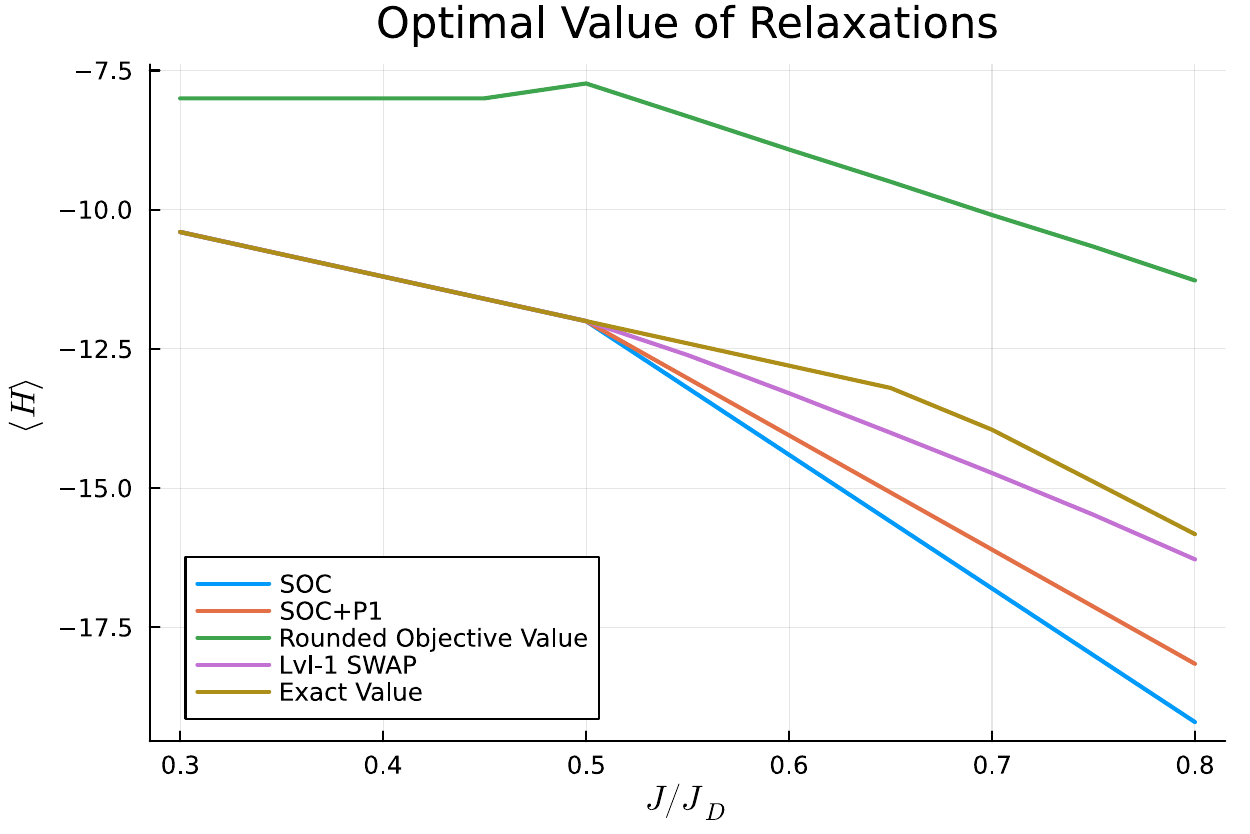}
\end{subfigure}
\begin{subfigure}{0.49\textwidth}
\centering
 \includegraphics[width= 0.75\textwidth]{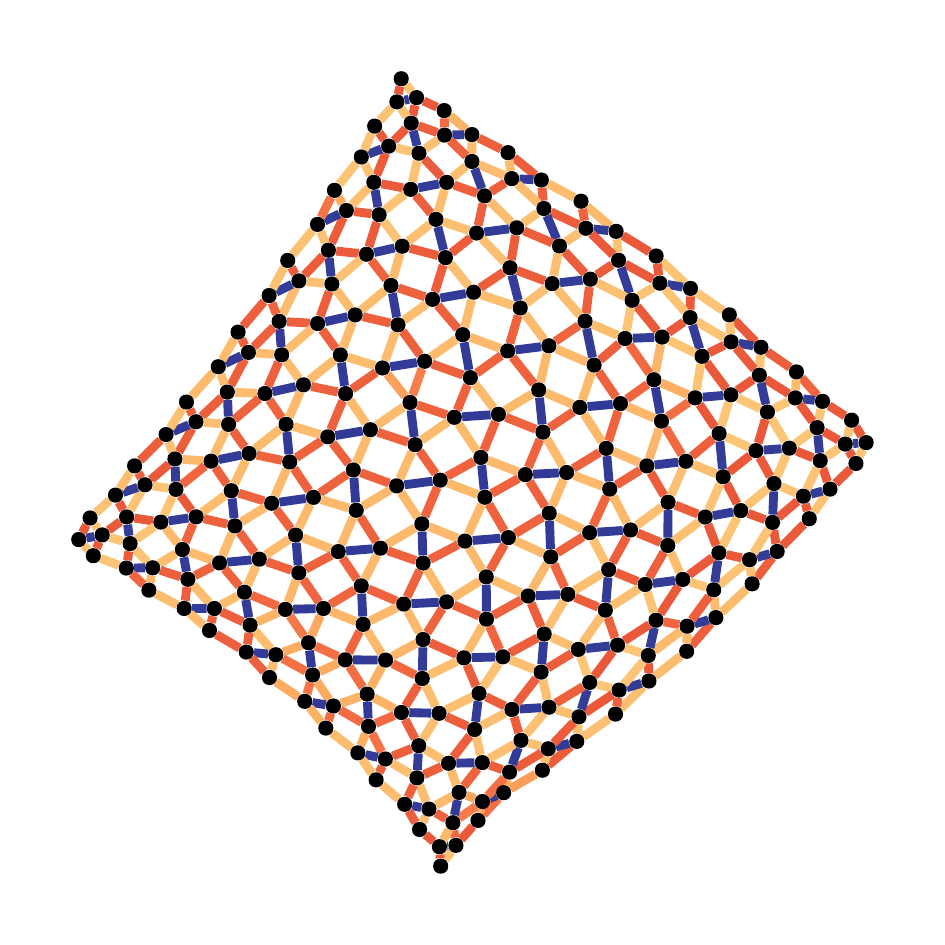}
\end{subfigure}
 \caption{
 \label{fig:SS_opt}
Optimal solutions of our relaxations for Shastry-Sutherland models with scaling as in \Cref{eq:qmc_hamil}.  SOC corresponds to \Cref{eq:relax2} and SOC+P1 corresponds to \Cref{eq:Pauli1}.  On the Left we plot the optimal solutions versus the exact ground state for a lattice of size $16$ qubits where $\langle H \rangle$ indicates the optimal objective whether it be the relaxation or exact solution.  On the right a heat map arising from the solution to a $256$ qubit disordered instance with parameters $J/J_D=0.4$ and $\sigma=0.05$.  
 }
\end{figure}

\begin{figure}[tbp]
\includegraphics[width= 0.6\textwidth]{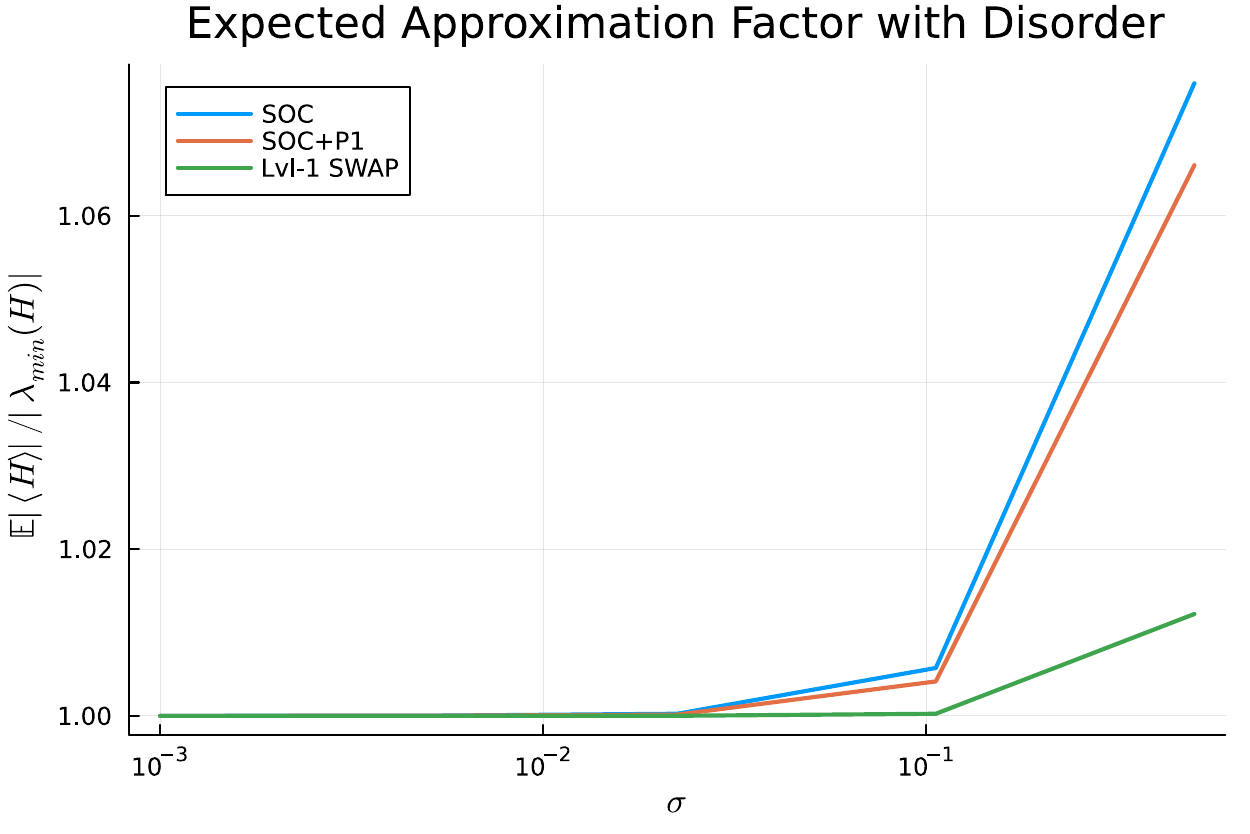}
 \caption{
 \label{fig:SS_disorder}
Expected approximation factor for Shastry-Sutherland model with disorder with scaling as in \Cref{eq:qmc_hamil}.  $\langle H \rangle$ indicates the optimal value for the SOCP (\Cref{eq:relax2}) or the SOCP suppplemented with the first level of the Pauli Hierarchy (\Cref{eq:Pauli1}).  $\lambda_{min}(H)$ is the ground state energy. All nonzero couplings in the model are multiplied by a factor of $(1+\sigma X)$ where $X$ is a normally distributed random variable.  Average ratio is computed with $1000$ samples.  The lattice of size size $16$.
 }
\end{figure}

\begin{figure}[tbp]
\includegraphics[width= 0.6\textwidth]{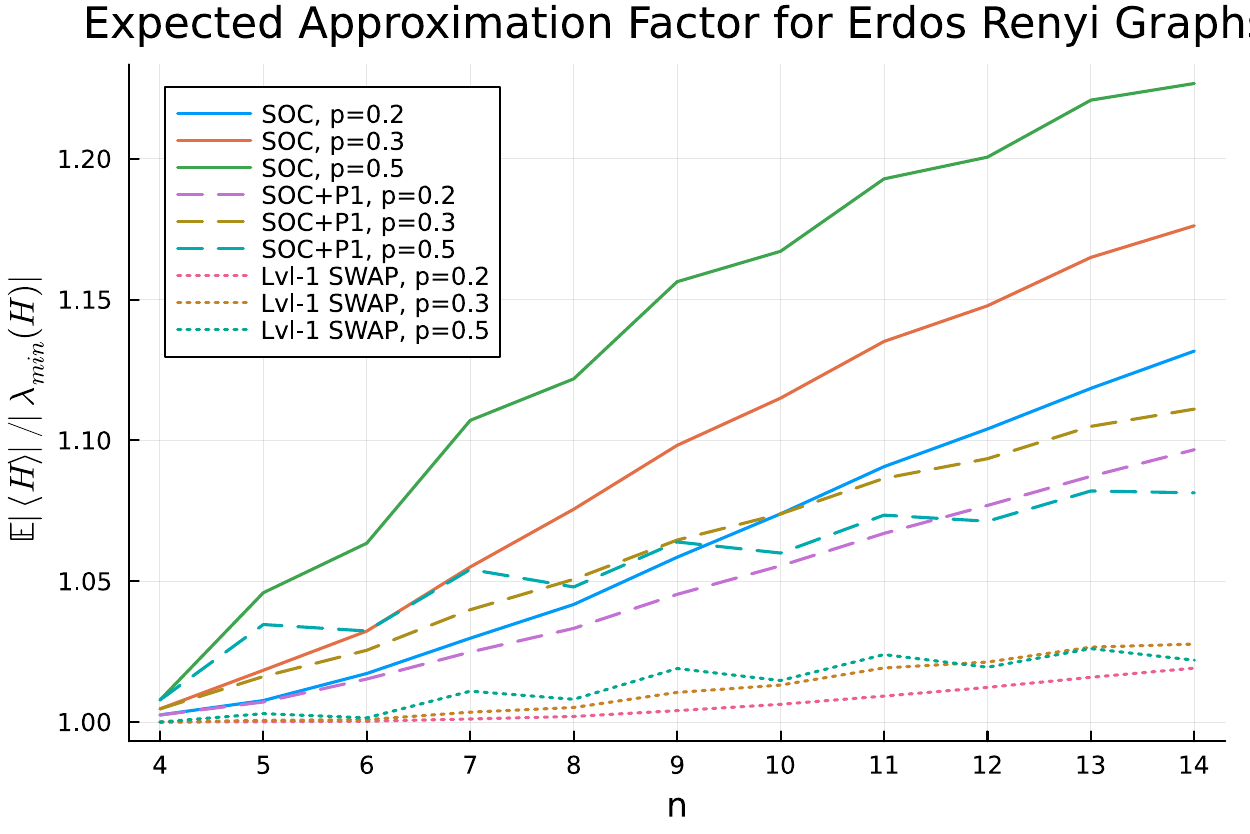}
 \caption{
 \label{fig:ER_objectives}
Expected approximation factor for Erdos-Renyi graphs of varying sizes with scaling as in \Cref{eq:qmc_hamil}.  $\langle H \rangle$ indicates the optimal value for the SOCP (\Cref{eq:relax2}) or the SOCP supplemented with the first level of the Pauli Hierarchy (\Cref{eq:Pauli1}). $\lambda_{min}(H)$ is the ground state energy.  The probability of edge formation is $p$ and the number of qubits in the graph is $n$.
 }
\end{figure}


\begin{table}
 \begin{tabular}{@{}c r r r r r@{}}
 \toprule
 n  & 0.2 & 0.4 & 0.6 & 0.8 & 1 \\
 \midrule
 10 & 1.03  & 1.12  & 1.31  & 1.70  & 3.00 \\
 12 & 1.04  & 1.16  & 1.38  & 1.92  & 3.67 \\
 14 & 1.05  & 1.19  & 1.44  & 2.13  & 4.33 \\
 16 & 1.06  & 1.21  & 1.52  & 2.31  & 5.00 \\
 18 & 1.07  & 1.23  & 1.58  & 2.47  & 5.67 \\
 20 & 1.08  & 1.24  & 1.64  & 2.64  & 6.33 \\
 25 & 1.09  & 1.28  & 1.80  & 2.97  & 8.00 \\
 30 & 1.11  & 1.32  & 1.92  & 3.27  & 9.66 \\
 35 & 1.13  & 1.36  & 2.02  & 3.5   & 11.33 \\
 40 & 1.15  & 1.42  & 2.14  & 3.72  & 13.00  \\
 \bottomrule
 \end{tabular}
 \caption{\label{fig:er_ratios}Average ratio of the value of the SOC / SOC+Pauli-1 relaxations
 for random Erdős–Rényi graphs with $n$ where an edge has probability $p$.
 We used 500 samples for $n\leq 20$; 100 samples for $n\leq 40$.
 In VarBench normalization.}
 \end{table}

\appendix
\section{Four qubits}
\label{app:four}
We discuss the state conditions for four-qubit marginals.

The analog of the relaxation in Eq.~\eqref{eq:relax1} to the quantum max cut problem~\eqref{eq:qmc_hamil}
using mutually consistent four-qubit marginals reads:
\begin{align}\label{eq:relax1_4rdm}
 \underset{\{\varrho_{ijk\ell}\}}{\min} \quad &-\frac{1}{2}\sum_{(i,j)\in E} \tr\big((\one - \swap_{ij}) \varrho_{ij}\big)\nn\\
 \text{subject to}        \quad &\varrho_{ij} = \tr_{k\ell} (\varrho_{ijk\ell})\,, \nn\\
                                &\varrho_{ijk\ell} \geq 0 \,\nn\\
                                &\tr(\varrho_{ijk\ell} ) = 1\,, \nn\\
                                & \tr_\ell(\varrho_{ijk\ell}) = \tr_m(\varrho_{ijkm})\,, \quad \quad \ell \neq m\,.
\end{align}
Above, the sum and constraints are over all pairwise different $i,j,k,\ell,m \in \{1,\dots, n\}$.
Denote by
\begin{align}
 x_{ij} = \tr\big(\varrho (ij) \big)\,, \quad\quad
 x_{ij,kl} = \tr\big(\varrho (ij)(kl) \big)\,.
\end{align}
Due to the identity in Eq.~\eqref{eq:3identity},
all three-body correlations can be expressed in terms of two-body correlations,
\begin{equation}
 \langle (ijk) \rangle + \langle (ikj) \rangle =
 \langle (ij) \rangle +
 \langle (ik) \rangle +
 \langle (jk) \rangle - 1\,.
\end{equation}

This implies that the constraint
$\tr_\ell(\varrho_{ijk\ell}) = \tr_m(\varrho_{ijkm})$
for $\ell \neq m$ is mediated through the $X_{ij}$ only,
that is the only variables shared by two different marginals are contained in the $x_{ij}$.

\subsection{Four-qubit identities}
We now reduce the number of variables further.
An analysis of the Gröbner basis given in Ref.~\cite[Appendix C]{Watts2024relaxationsexact} shows,
in addition to the identity $(123) + (123)^{-1} = (12) + (23) + (13) - \id$~\footnote{
In fact all identities of $S_n$ acting on $(\C^2)^{\otimes n}$
are {\em generated} by $\id - (12) - (23) - (13) + (123) + (123)^{-1} = 0$.}
that the remaining identities for $S_4$ on the subspace of matrices in $\Sigma_{4,2}^{\sym}$
are given by
\begin{equation}\label{eq:S4_identity}
 (1234) + (1234)^{-1} = -\id + (13) + (24) + (12)(34) + (14)(23) - (13)(24)\
\end{equation}
and permutations thereof. See also the analytical derivation in
Ref.~\cite{procesi2024specialbases2swapalgebras}.
Together with the three-qubit identity~in Eq.~\eqref{eq:3identity},
this shows that the set of involutions
of the form $(i,j)$ and $(i,j)(k,l)$ span $\Sigma_{4,2}^{\sym}$.
In fact, a similar result holds for $n$-partite systems:
Procesi has shown that
the space $\Sigma_{n,2}^{\sym}$ is spanned by involutions~\cite{procesi2024specialbases2swapalgebras}.

\subsection{State conditions}

We are now in position to derive the analogue of Eq.~\eqref{eq:conditions_red},
that is the conditions for an operator to be an real invariant four-qubit state.
With Eq.~\eqref{eq:S4_identity}, consider the positivity condition in Eq.~\eqref{eq:pos_in_irrep}
for an invariant operator $A \in \Sigma_{4,2}^{\sym}$
\begin{equation}\label{eq:4qb}
 P_\lambda \sum_{\sigma \in S_k} \tr(\sigma^{-1} A) R_\lambda( \sigma) \succeq 0
\end{equation}
for the irreps $\lambda = (4)$, $(3,1)$, and $(2,2)$.
Appendix~\ref{app:irreps} lists orthogonal representations of $S_4$. Then the three central Young projectors expand as [see Eq.~\eqref{eq:young}]

\begin{align}
 P_{(4)} &= \frac{1}{24} \sum_{\pi \in S_4} \pi \,,
 \nn\\
 P_{(31)} &= \frac{1}{8}
 \Big(3\id + (12) + (13) + (14) + (23) + (24) + (34) - (12)(34) - (13)(24) - (14)(23)  \nn\\
 &\quad\quad - (1234) - (1243) - (1342) - (1324) - (1432) - (1423)
 \Big)\,,
 \nn\\
 P_{(22)} &= \frac{1}{12} \Big(
 2\big( \id + (13)(24) + (14)(23) + (12)(34)\big) \nn\\
 &\quad\quad - (234) - (243) - (123) - (124) - (132) - (134) - (142) - (143)
 \Big) \,. \nn
\end{align}
On ~$(\C^2)^{\otimes 4}$ we can variable reduce
through the identities in Eq.~\eqref{eq:3identity} and \eqref{eq:4qb},
\begin{align}
 P_{(4)} &= \frac{1}{12}\Big(
 -3\id
 + 2\, \big[ (12) + (13) + (14) + (23) + (24) + (34)\big]
 + (12)(34) + (13)(24) + (14)(23)  \Big) \,,
 \nn\\
 P_{(31)} &= \frac{1}{4}\Big(
 3\id - (12)(34) - (13)(24) - (14)(23)
 \Big)\,,
 \nn\\
 P_{(22)} &= \frac{1}{6}\Big(
 3\id - (12) - (13) -  (14) - (23) - (24) - (34) + (12)(34) + (13)(24)  + (14)(23)
 \Big)
 \,. \nn
\end{align}

Denote their expectation values by
\begin{align}\label{eq:r_lambda}
 r_{(4)}   &= \tr( P_{(4)}   \varrho )\,, &
 r_{(3,1)} &= \tr( P_{(3,1)} \varrho )\,, &
 r_{(2,2)} &= \tr( P_{(2,2)} \varrho )\,.
\end{align}

Then the normalization condition $\tr(\varrho) = 1$
requires that
\begin{equation}
 r_{(4)}, r_{(3,1)}, r_{(2,2)} \geq 0\,,\quad\quad  r_{(4)} +  r_{(3,1)} + r_{(2,2)} = 1\,.
\end{equation}

The positive semidefinite constraint can be obtained
through Proposition~\ref{eq:pos_in_irrep_simple}
and using the irreducible orthogonal representations given in Appendix~\ref{app:irreps}.
The conditions for $A \succeq 0$, $A \in \Sigma_{4}^+$ are, listed by partitions:

{

\begin{equation}
-3
+ 2\, \big[ x_{12} + x_{13} + x_{14} + x_{23} + x_{24} + x_{34} \big]
+ x_{12,34} + x_{13,24} + x_{14,23} \geq 0\,; \tag{Partition [4]}
\end{equation}

\begin{equation}\tag{Partition [3,1]}
\begin{pmatrix}
 A_{00} &A_{01} &A_{02} \\
 A_{01} &A_{11} &A_{12} \\
 A_{02} &A_{12} &A_{22}
\end{pmatrix} \geq 0\,,
\end{equation}
where
\begin{align}
 A_{00} &=
 \frac{2}{3} \big(
3
+ 2\, \big[
  x_{12}
+ x_{13}
- x_{14}
+ x_{23}
- x_{24}
- x_{34} \big]
- x_{12,34}
- x_{13,24}
- x_{14,23}
\big)
 \,,\nn\\
 A_{01} &=
 \frac{\sqrt{2}}{3}\,
 \big(
 -2 \, x_{12}
 + x_{13}
 - x_{14}
 + x_{23}
 - x_{24}
 + 2 \, x_{34}
 + 4 \, x_{12,34}
 - 2 \, x_{13,24}
 - 2 \, x_{14,23}
 \big)
 \,,\nn\\
 A_{02} &=
 \frac{\sqrt{6}}{3}\,\big(
- x_{13}
- x_{14}
+ x_{23}
+ x_{24}
+ 2 \, x_{13,24}
- 2 \, x_{14,23}
\big)
 \,,\nn\\
 A_{11} &=
 \frac{2}{3} \,\big(
 3
 + x_{12}
 - x_{34}
 + 2 \, x_{24}
 - 2 \, x_{13}
 + 2 \, x_{14}
 - 2 \, x_{23}
 + x_{12,34}
 - 2 \, x_{13,24}
 - 2 \, x_{14,23}
  \big)
 \,,\nn\\
 A_{12} &=
 \frac{2}{\sqrt{3}}\,\big(
 -x_{13} - x_{14} + x_{23} + x_{24} - x_{13,24} + x_{14,23}
 \big)
 \,,\nn\\
 A_{22} &= 2\,\big(
 1 - x_{12} - x_{12,34} + x_{34}
 \big)
 \,; \nn
\end{align}
and
\begin{equation}\label{eq:partition_31}\tag{Partition [2,2]}
 \begin{pmatrix}
  B_{00} & B_{01} \\
  B_{01} & B_{11}
 \end{pmatrix}
 \geq 0\,,
\end{equation}
where
\begin{align}
B_{00} &= 3
+ x_{12}
- 2 \, x_{13}
- 2 \, x_{14}
- 2 \, x_{23}
- 2 \, x_{24}
+ x_{34}
- x_{12,34}
+ 2 \, x_{13,24}
+ 2 \, x_{14,23} \big)
\,,\nn\\
B_{01} &= \sqrt{3}\big(
- x_{13} + x_{14} + x_{23} - x_{24} + x_{13,24} - x_{14,23}
\big)
\,,\nn\\
B_{11} &= 3\, \big(1 - x_{12} - x_{34} + x_{12,34} \big)\,.
\end{align}
}
These correspond to the fact that a state is positive semidefinite on the irreducible representations corresponding to
$(4), (3,1)$, and $(2,2)$ respectively.

\subsection{Second order cone formulation}
As in the three-qubit case, the condition of \eqref{eq:partition_31}
can be formulated as a second-order cone constraint.
The corresponding subnormalized real Bloch ball reads:
\begin{equation}
 B = \begin{pmatrix}
      b_0 + b_1 & b_2 \\
      b_2          & b_0 - b_1 \\
     \end{pmatrix} \succeq 0\,.
\end{equation}
with $b_0 = (B_{00} + B_{11})/2 $, $b_1 = (B_{00} - B_{11})/2$, $b_2 = B_{01}$. Then
\begin{align}
b_0 &=\frac{1}{6}\Big( 3 - x_{12} - x_{13} - x_{14} - x_{23} - x_{24} - x_{34} + x_{12,34} + x_{13,24}  + x_{14,23} \Big) \,,\nn\\
b_1 &=  2\,x_{12} - 2\,x_{12,34} - x_{13} + x_{13,24} - x_{14}+ x_{14,23} - x_{23} - x_{24} + 2\,x_{34} \nn\nn\\
b_2 &= \sqrt{3}\big( - x_{13} + x_{14} + x_{23} - x_{24} + x_{13,24} - x_{14,23}
\big) \,.
\end{align}
We use the fact that
 \begin{equation}\label{eq:SOC_bloch_B}
  \begin{pmatrix}
   b_0 + b_1 & b_2 \\
   b_2 & b_0 - b_1
  \end{pmatrix} \geq 0
  \quad\quad\iff\quad\quad
  b_1^2 +  b_2^2 \leq b_0^2,\quad b_0+b_1\geq 0,\quad b_0 - b_1\geq 0\,.
 \end{equation}

\begin{remark}
Also the $3\times 3$ semidefinite constraint from the partition $[3,1]$ can be relaxed to a second-order cone constraint
by imposing that all of its $2\times 2$ minors are positive semidefinite.
This does not seem to give a computational speed-up.
\end{remark}

\section{Irreducible Representations of $S_3$ and $S_4$}
\label{app:irreps}

We list the orthogonal irreducible representations $R_\lambda$ of $S_3$ and $S_4$ below.
These are needed to obtain the decomposition in terms of $x_{ij}$ and $x_{ij,kl}$ in Appendix~\ref{app:four}.

\noindent {\bf $\mathbf S_3$.} The symmetric group $S_3$ has three irreps labeled by partitions $[3]$, $[2,1]$, and $[1,1,1]$.
Of these, the first two are non-zero while the sign representation $[1,1,1]$ vanishes on $(\C^2)^{\ot 3}$.

\bigskip
\begin{tabular}{@{}lcccccc@{}}
 \toprule
 $\lambda$ & $\id$ & $(12)$     & $(23)$       & $(13)$       & $(123)$     & $(132)$ \\
 \midrule
 $[3]$ & $1$   & $1$        & $1$          & $1$          & $1$           & $1$ \\
 $[2,1]$
 & $\begin{pmatrix}
    1 & 0 \\
    0 & 1
   \end{pmatrix}$
 & $\begin{pmatrix}
    1 & 0 \\
    0 & -1
   \end{pmatrix}$
 &
 $\tfrac{1}{2}\begin{pmatrix}
   -1 & \sqrt{3} \\     \sqrt{3} &  1
   \end{pmatrix}$
&
$\tfrac{1}{2}\begin{pmatrix}
   -1 & -\sqrt{3} \\   -\sqrt{3} &  1
  \end{pmatrix}$
&
$\tfrac{1}{2}\begin{pmatrix}
   -1 & \sqrt{3}   \\   -\sqrt{3} & -1
  \end{pmatrix}$
&
$\tfrac{1}{2}\begin{pmatrix}
   -1 &  -\sqrt{3} \\    \sqrt{3} & -1
  \end{pmatrix}$ \\
$[1,1,1]$ & $1$   & $-1$        & $-1$          & $-1$          & $1$           & $1$ \\
  \bottomrule
\end{tabular}

\bigskip

\noindent {\bf $\mathbf S_4$.}
The symmetric group $S_4$ has irreducible representation labeled by the partitions
$[4]$, $[3,1]$, $[2,2]$, $[2,1,1]$, $[1,1,1,1]$.
Of these, by the Schur-Weyl duality only $[4]$, $[3,1]$, and $[2,2]$ are nonvanishing, these having no more than $2$ rows.
The completely symmetric irrep $[4]$ is trivial, having value $R_{[4]}(\pi) = 1$ for all $\pi \in S_4$.
We list the remaining nonvanishing irreps below.
Again we can do variable reduction, and
Appendix B in Ref.~\cite{Watts2024relaxationsexact} together with the fact that Quantum Max Cut is real
shows that the variables $(i,j)$ and $(i,j)(k,l)$ span the space of permutations on $(\C^d)^{\ot 4}$.

\renewcommand{\arraystretch}{1.4}
\begin{xltabular}{\linewidth}{@{}lXl@{}}
\toprule
Partition & Permutation (all after $\ast$ redundant on $(\C^2)^{\ot 4}$) & Irreducible Representation \\
%
%
\midrule
$[3,1]$
& $ \id $; &
$ \phantom{-\frac{1}{2}}
\begin{pmatrix}
1 & 0 & 0 \\
0 & 1 & 0 \\
0 & 0 & 1
\end{pmatrix} $
\\
& $ (13)(24) $; &
$ -\frac{1}{3}
\begin{pmatrix}
1 & \sqrt{2} & -\sqrt{6} \\
\sqrt{2} & 2 & \sqrt{3} \\
-\sqrt{6} & \sqrt{3} & 0
\end{pmatrix} $
\\
& $ (14)(23) $; &
$ -\frac{1}{3}
\begin{pmatrix}
1 & \sqrt{2} & \sqrt{6} \\
\sqrt{2} & 2 & -\sqrt{3} \\
\sqrt{6} & -\sqrt{3} & 0
\end{pmatrix} $
\\
& $ (12)(34) $; &
$ \phantom{-}\frac{1}{3}
\begin{pmatrix}
-1 & 2 \, \sqrt{2} & 0 \\
2 \, \sqrt{2} & 1 & 0 \\
0 & 0 & -3
\end{pmatrix} $
\\
& $\ast \,\, (234) $; $ (243) $; &
$ \phantom{-}\frac{1}{6}
\begin{pmatrix}
-2 & \sqrt{2} & \sqrt{6} \\
\sqrt{2} & -1 & 2 \, \sqrt{3} \\
\sqrt{6} & 2 \, \sqrt{3} & 3
\end{pmatrix} $
\\
& $\ast \,\, (132) $; $ (123) $; &
$ \phantom{-}\frac{1}{2}
\begin{pmatrix}
2 & 0 & 0 \\
0 & -1 & 0 \\
0 & 0 & -1
\end{pmatrix} $
\\
& $\ast \,\, (143) $; $ (134) $; &
$ - \frac{1}{6}
\begin{pmatrix}
2 & -\sqrt{2} & \sqrt{6} \\
-\sqrt{2} & 1 & 2 \, \sqrt{3} \\
\sqrt{6} & 2 \, \sqrt{3} & -3
\end{pmatrix} $
\\
& $\ast \,\, (124) $; $ (142) $; &
$ \phantom{-}\frac{1}{6}
\begin{pmatrix}
-2 & -2 \, \sqrt{2} & 0 \\
-2 \, \sqrt{2} & 5 & 0 \\
0 & 0 & -3
\end{pmatrix} $
\\
& $ (34) $; &
$ \phantom{-}\frac{1}{3}
\begin{pmatrix}
-1 & 2 \, \sqrt{2} & 0 \\
2 \, \sqrt{2} & 1 & 0 \\
0 & 0 & 3
\end{pmatrix} $
\\
& $\ast \,\,  (1324) $; $ (1423) $; &
$ -\frac{1}{3}
\begin{pmatrix}
1 & \sqrt{2} & 0 \\
\sqrt{2} & 2 & 0 \\
0 & 0 & 0
\end{pmatrix} $
\\
& $ (12) $; &
$ \phantom{-\frac{1}{2}}
\begin{pmatrix}
1 & 0 & 0 \\
0 & 1 & 0 \\
0 & 0 & -1
\end{pmatrix} $
\\
& $ (23) $; &
$ \phantom{-}\frac{1}{2}
\begin{pmatrix}
2 & 0 & 0 \\
0 & -1 & \sqrt{3} \\
0 & \sqrt{3} & 1
\end{pmatrix} $
\\
& $\ast \,\, (1342) $; $ (1243) $; &
$ \phantom{-}\frac{1}{6}
\begin{pmatrix}
-2 & \sqrt{2} & \sqrt{6} \\
\sqrt{2} & -1 & -\sqrt{3} \\
\sqrt{6} & -\sqrt{3} & -3
\end{pmatrix} $
\\
& $ (14) $; &
$ -\frac{1}{6}
\begin{pmatrix}
2 & 2 \, \sqrt{2} & 2 \, \sqrt{6} \\
2 \, \sqrt{2} & -5 & \sqrt{3} \\
2 \, \sqrt{6} & \sqrt{3} & -3
\end{pmatrix} $
\\
& $ (24) $; &
$ \phantom{-}\frac{1}{6}
\begin{pmatrix}
-2 & -2 \, \sqrt{2} & 2 \, \sqrt{6}\\
-2 \, \sqrt{2} & 5 & \sqrt{3} \\
2 \, \sqrt{6}& \sqrt{3} & 3
\end{pmatrix} $
\\
& $ (13) $; &
$ \phantom{-}\frac{1}{2}
\begin{pmatrix}
2 & 0 & 0 \\
0 & -1 & -\sqrt{3} \\
0 & -\sqrt{3} & 1
\end{pmatrix} $
\\
& $\ast \,\,  (1432) $; $ (1234) $; &
$ -\frac{1}{6}
\begin{pmatrix}
2 & -\sqrt{2} & \sqrt{6} \\
-\sqrt{2} & 1 & -\sqrt{3} \\
\sqrt{6} & -\sqrt{3} & 3
\end{pmatrix} $
\\
%
%
\midrule
$[2,2]$
& $ \id $; $ (13)(24) $; $ (14)(23) $; $ (12)(34) $; &
$
\phantom{-\frac{1}{2}}
\begin{pmatrix}
1 & 0 \\
0 & 1
\end{pmatrix} $
\\
& $\ast \,\, (234) $; $ (132) $; $ (143) $; $ (124) $; $ (243) $; $ (134) $; $ (142) $; $ (123) $; &
$ -\frac{1}{2}
\begin{pmatrix}
1 & 0 \\
0 & 1
\end{pmatrix} $
\\
& $ (34) $; $ (12) $;  $ \ast \,\, (1324) $; $ (1423) $; &
$ \phantom{-\frac{1}{2}}
\begin{pmatrix}
1 & 0 \\
0 & -1
\end{pmatrix} $
\\
& $ (23) $; $ (1342) $; $ (14) $; $ (1243) $; &
$
\phantom{-}\frac{1}{2}
\begin{pmatrix}
-1 & \sqrt{3} \\
\sqrt{3} & 1
\end{pmatrix} $
\\
& $ (24) $; $ (13) $; $ \ast \,\,  (1432) $; $ (1234) $; &
$
\phantom{-}\frac{1}{2}
\begin{pmatrix}
-1 & -\sqrt{3} \\
- \sqrt{3} & 1
\end{pmatrix} $
\\
%
%
%
\bottomrule
\end{xltabular}

\bibliographystyle{alpha}

\bibliography{current_bib}
\end{document}